\documentclass{article}
\usepackage{stmaryrd,amssymb,amsmath,amsthm,url}
\usepackage{graphicx}
\title{On Hoare-McCarthy Algebras}

\author{
	Jan A.\ Bergstra\thanks{J.A. Bergstra acknowledges
	support from NWO (project Thread Algebra for Strategic
	Interleaving).} \and
	Alban Ponse \\
\\
  {\small
	  Section Theory of Computer Science,
	  Informatics Institute,
	  University of Amsterdam}\\
	{\small Url: 
	\url{www.science.uva.nl/~{janb,alban}}
	}
}
\date{}

\newcommand{\RP}{\ensuremath{\mathit{RP}}}
\newcommand{\WM}{\ensuremath{\mathit{WM}}}
\newcommand{\lcon}{\ensuremath{\leadsto}}

\newcommand{\CPS}{\axname{{CTS}}}
\newcommand{\CCE}{\ensuremath{{CceTh}}}
\newcommand{\HMA}{\ensuremath{\mathbf{HMA}}}
\newcommand{\BA}{\ensuremath{\mathbb{A}}}
\newcommand{\NT}{\ensuremath{\mathcal{T}}}
\newcommand{\oT}{\ensuremath{\mathbb{T}}}
\newcommand{\SigmaStateless}{\ensuremath{\Sigma^A_{\textit{ce}}}}
\newcommand{\SigmaState}{\ensuremath{\Sigma^A_{\textit{sce}}}}
\newcommand{\SigmaStateful}{\ensuremath{\Sigma^A_{\textit{spa}}}}
\newcommand{\SigmaHMA}{\ensuremath{\Sigma^A_{\textit{ce}}}}

\newcommand{\leftand}{~
     \mathbin{\setlength{\unitlength}{1ex}
     \begin{picture}(1.4,1.8)(-.3,0)
     \put(-.6,0){$\wedge$}
     \put(-.53,-0.36){\circle{0.6}}
     \end{picture}
     }}
     
\newcommand{\fr}{\ensuremath{{fr}}}
\newcommand{\rp}{\ensuremath{{rp}}}
\newcommand{\con}{\ensuremath{{cr}}}
\newcommand{\wmem}{\ensuremath{{wm}}}
\newcommand{\mem}{\ensuremath{{mem}}}
\newcommand{\stat}{\ensuremath{{st}}}
\newcommand{\tr}{\ensuremath{T}}
\newcommand{\fa}{\ensuremath{F}}
\newcommand{\lef}{\ensuremath{\triangleleft}}
\newcommand{\rig}{\ensuremath{\triangleright}}

\newcommand{\NC}{\ensuremath{{\mathcal C}}}
\newtheorem{theorem}{Theorem}
\newtheorem{lemma}{Lemma}  
\newtheorem{proposition}{Proposition}  
  
\newtheorem{corollary}{Corollary}  
\newtheorem{remark}{Remark}  
  
\newtheorem{definition}{Definition}  
\theoremstyle{definition}
\newtheorem*{la}{Lemma}

\newcommand{\CP}{\axname{CP}}

\newcommand{\NS}{\ensuremath{\mathbb S}}
\newcommand{\Nplus}{\ensuremath{\mathbb N}^+}

\newcommand{\axname}[1]{\textup{\ensuremath{\textrm{#1}}}}

\newcommand{\apply}{\bullet}
\newcommand{\reply}{\:\mathbf{!}\:}

\begin{document}

\maketitle

\begin{abstract}
We discuss an algebraic approach to propositional logic with side effects. 
To this end, we
use Hoare's conditional [1985], which is a ternary connective
comparable to if-then-else.
Starting from McCarthy's notion of sequential
evaluation [1963] we discuss a number of valuation
congruences and we introduce Hoare-McCarthy algebras
as the structures that characterize these congruences.
\end{abstract}

{\small \tableofcontents}

\section{Introduction}
In the paper~\cite{BP10} we introduced 
\emph{proposition algebra}, an account of propositional
logic with side effects in an algebraic, equational style.
We define several semantics, all of which identify less than
conventional propositional logic (PL), and the one that 
identifies least is named \emph{free valuation congruence}.

Free valuation congruence can be roughly explained 
as follows: consider
valuation functions defined on strings of propositional 
variables (atoms), then two propositional 
statements $P$ and $Q$ are free valuation equivalent 
if under all such
valuations they yield the same Boolean value, i.e., either
\tr\ (true) or \fa\ (false). For example, the associativity 
of conjunction is preserved under free valuation equivalence,
and $P\wedge\fa$ is free valuation equivalent with \fa\
(both evaluate to \fa).
However, free valuation
equivalence is not a congruence: continuing the
last example and assuming
evaluation proceeds from left to right and $a$ and $b$
are atoms,
\[(a\wedge\fa)\vee b
\quad\text{and}\quad \fa\vee b\]
yield different evaluation results
for any valuation function $f$ with $f(b)=\tr$
and $f(ab)=\fa$ because irrespective of the value of $f(a)$,
$(a\wedge\fa)$ evaluates under $f$ to \fa, and the evaluation
of $b$ in $(a\wedge\fa)\vee b$ is then determined by $f(ab)$,
while $\fa\vee b$ yields under valuation $f$ the value $f(b)$.
The requirement that
propositional statements are equal only if in each context
they yield the same value indeed admits the possibility 
to model side effects. Free valuation \emph{congruence},
defined as the largest
congruence contained in free valuation
equivalence, identifies less than free valuation equivalence
and is the semantical notion we are interested in.
As an example, associativity of conjunction is 
preserved under free valuation congruence.
So, in free valuation equivalence, 
the evaluation of an atom in a propositional statement 
depends on the evaluation history (i.e., the
atoms previously evaluated in that statement).
Although we failed to find a precise definition of a 
``side effect'', we 
use as a working hypothesis
that this kind of dependency models the occurrence of
side effects. 

As implied above,
the \emph{order} of evaluation is crucial in 
proposition algebra.
This immediately implies that the conventional connectives
$\wedge$ and $\vee$ are not appropriate 
because their symmetry is lost:
while in PL the propositional statements
\[\fa \wedge P\quad\text{and}\quad P\wedge\fa\]
are identified, they are not free valuation congruent: 
if evaluation proceeds from left to
right, the evaluation of $P$ in $P\wedge \fa$ may yield
a side effect that is not created upon the evaluation
of $\fa \wedge P$ (in the latter $P$ is not evaluated,
although both statements evaluate to \fa).

A logical connective that incorporates a fixed order
of evaluation ``by nature'' is Hoare's ternary connective 
\[x\lef y\rig z,\]
introduced in the paper~\cite{Hoa85}
as the \emph{conditional}.\footnote{Not to be 
 confused with Hoare's \emph{conditional}
 introduced in in his 1985~book on CSP~\cite{Hoa85a} 
 and in his well-known 1987 paper \emph{Laws of Programming}
 \cite{HHH87} 
 for expressions $P\lef b\rig Q$ with $P$ and $Q$ programs and 
 $b$ a Boolean expression;  these sources do
 not refer to~\cite{Hoa85} that appeared in 1985.}
A more common expression for the conditional $x\lef y\rig z$
is
\[
\emph{if $y$ then $x$ else $z$}
\]
with $x$, $y$ and $z$ ranging over 
propositional statements. 
However, in order to reason 
systematically with conditionals, a notation
such as $x\lef y\rig z$ seems indispensable, and equational
reasoning appears to be the most natural and elegant
type of reasoning.
Note that a left-sequential conjunction $x\wedge y$
can be 
expressed as $y\lef x\rig\fa$. In this paper we 
restrict to the conditional as the only primitive
connective; in the papers~\cite{BP10,BP10a}
we use the notation ${\leftand}$ (taken from~\cite{BBR95})
for left-sequential conjunction
and elaborate on the connection between sequential
binary connectives and the conditional; we return to this 
point in our conclusions (Section~\ref{sec:Con}). 
In~\cite{Hoa85}, Hoare proves that propositional 
logic can be equationally
characterized over the signature
$\Sigma_{\CP}=\{\tr,\fa,\_\lef\_\rig\_\}$
and provides a set of elegant axioms to this end,
including those in Table~\ref{tab:CP}. 

\begin{table}
\centering
\hrule
\begin{align*}
\label{CP1}
\tag{CP1} x \lef \tr \rig y &= x\\
\label{CP2}\tag{CP2}
x \lef \fa \rig y &= y\\
\label{CP3}\tag{CP3}
\tr \lef x \rig \fa  &= x\\
\label{CP4}\tag{CP4}
\qquad
    x \lef (y \lef z \rig u)\rig v &= 
	(x \lef y \rig v) \lef z \rig (x \lef u \rig v)
\end{align*}
\hrule
\caption{The set \CP\ of axioms for proposition algebra}
\label{tab:CP}
\end{table}

In \cite{BP10} we 
define varieties of so-called
\emph{valuation algebras} in order to provide a 
semantic framework for proposition algebra.
These varieties serve the interpretation of a logic
over $\Sigma_{\CP}$ by means of sequential 
evaluation:
in the evaluation of $t_1\lef t_2\rig t_3$, first
$t_2$ is evaluated, and the result of this evaluation
determines further evaluation; upon $\tr$,
$t_1$ is evaluated and determines the final
evaluation result ($t_3$ is not evaluated); upon
\fa, $t_3$ is evaluated and determines the final
evaluation result 
($t_1$ is not evaluated).\footnote{Sequential
  evaluation is also called
  \emph{short-circuit}, \emph{minimal} or
  \emph{McCarthy evaluation}, and
  can be traced back to McCarthy's seminal
  paper~\cite{McC63}.}
The interpretation of propositional statements
that is defined by each of the varieties discussed 
in~\cite{BP10} 
satisfies the axioms in Table~\ref{tab:CP}, and
the interpretation of propositional statements
defined by
the most distinguishing variety is axiomatized by
exactly these four axioms. We
write \CP\ for this set of axioms (where CP 
abbreviates conditional propositions) and
$=_\fr$ (free valuation congruence)
for the associated valuation
congruence. Thus for each pair of closed terms
$t,t'$ over $\Sigma_{\CP}$, i.e., 
terms that do not contain 
variables, but that 
of course may contain atoms (propositional
variables),
\begin{equation}
\label{id:een}
\CP\vdash t=t'\iff t=_\fr t'.
\end{equation}
In~\cite{Chris} it is shown that \CP\ is an
independent axiomatization, and also that \CP\ is 
$\omega$-complete if the set $A$ of 
atoms involved contains at least two elements.
A further introduction to the semantics
defined in~\cite{BP10} can be found 
in Section~\ref{sec:Con}.

In this paper we provide an alternative semantics
for proposition algebra.
We define a particular
type of two-sorted algebras that capture both 
axiomatic derivability and semantic congruence 
at the same time. We call these algebras 
\emph{Hoare-McCarthy algebras} 
(HMAs for short) and for a number of valuation
congruences we prove the existence of a `canonical'
HMA in which axiomatic derivability and semantic
congruence coincide. Thus, our first typical result 
is
\begin{equation}
\label{id:twee}
\CP\vdash t=t'\iff \BA^{sc}\models t=t',
\end{equation}
where $\BA^{sc}$ is the canonical HMA referred to above. 
Here the
direction $\Longrightarrow$ indicates soundness of 
the axiom set \CP\ (which appears to hold in each HMA),
and the other direction indicates completeness.
Thus, the combination of \eqref{id:een} en \eqref{id:twee}
shows that we can characterize free valuation congruence
in a single HMA. A further discussion about the semantics
defined in~\cite{BP10} and the semantics defined in 
this paper and a comparison of these
can be found in Section~\ref{sec:Con}.

In Sections~\ref{sec:rp}-\ref{sec:stat} we consider
classes of valuation functions defined on 
(subsets of) $A^+$
with the property that 
\[f(a_1...a_na_{n+1})\in\{\tr,\fa\}\]
gives the reply of valuation $f$ on atom $a_{n+1}$
after $a_1$ up to $a_n$ have been evaluated, so both
$a_{n+1}$ and the valuation 
history $a_1... a_n$ determine the result of evaluation.
The class of all valuation
functions defines \emph{structural congruence}
(which coincides with free valuation congruence),
these function all have
domain $A^+$ (each valuation history is
significant), and the class of valuation functions
that defines \emph{static} congruence only
considers functions that have
$A$ as their domain (no valuation history is significant; 
this is equivalent to PL).
For $|A|>1$, domains 
that are strictly in between these two are
$A^{\con}$, the set of strings in which no atom has the 
same neighbour, 
and $A^{core}$, the set of strings in which each atom
occurs at most once. Note that if $A$ is finite,
$A^{\con}$ is infinite and
$A^{core}$ is finite, and if
$A=\{a\}$ then $A^{\con}=A^{core}=A$.
We define \emph{contractive} congruence using
$\{\tr,\fa\}^{A^{\con}}$ as its class of valuation functions,
and \emph{memorizing} congruence
with help of
$\{\tr,\fa\}^{A^{core}}$. We distinguish two more
congruences: \emph{repetition-proof}
congruence which is based on a subset of
the function space $\{\tr,\fa\}^{A^+}$, and \emph{weakly memorizing}
congruence which is based on a subset of 
the function space $\{\tr,\fa\}^{A^{core}}$.
For all congruences mentioned, we provide
complete axiomatizations, and in Section~\ref{sec:Con}
we relate these results to similar results proved
in~\cite{BP10}.

In some forthcoming definitions and proofs we use 
the empty string, which we always denote by $\epsilon$.
Furthermore, we use $\equiv$ to denote syntactic equivalence.

\section{Proposition algebras and HMAs}
\label{sec:1}

In this section we define proposition algebras and Hoare-McCarthy 
algebras.

Throughout this paper let $A$ be a non-empty,
denumerable set of atoms
(propositional variables). 
Define $C$ as the sort 
of conditional expressions with signature 
\[\SigmaStateless=\{a:C,~\tr:C,~\fa:C,~
.\lef.\rig.:C\times C\times C\rightarrow C\mid a\in A\},\]
thus each atom in $A$
is a constant of sort $C$. In $\SigmaStateless$,
\emph{ce} stands 
for ``conditional expressions''. 
We write $\NT_{\SigmaStateless}$ for the set of closed terms
over $\SigmaStateless$, and $\oT_{\SigmaStateless}$ for the 
set of all terms. Given an expression $t_1\lef t_2\rig t_3$
we will sometimes refer to $t_2$ as the \emph{central 
condition}. We assume that conditional composition
satisfies the axioms in Table~\ref{tab:CP}. We refer to
this set of axioms with \CP.

\begin{definition}
\label{def:PA}
A $\SigmaStateless$-algebra is a 
\textbf{proposition algebra} 
if it is a model of \CP.
\end{definition}

A non-trivial initial algebra $I(\SigmaStateless,\CP)$
exists. This can be easily shown in the setting of term
rewriting~\cite{Terese}.
Directing all \CP-axioms from left
to right yields a strongly normalizing TRS (term rewriting
system) for closed terms: 
define
a weight fuction $w:\NT_{\SigmaStateless}\rightarrow 
\Nplus$ by
\begin{align*}
w(a)&=2\quad\text{for all }a\in A\\
w(\tr)&=2\\
w(\fa)&=2\\
w(x\lef y\rig z)&=(w(x)\cdot w(z))^{w(y)}
\end{align*}
Clearly, for all rewrite rules $l\rightarrow r$ and closed
substitutions $\sigma$ we have
$w(\sigma(l))>w(\sigma(r))$. 
It is also not difficult to see that
this TRS is weakly confluent, the
critical pairs $\langle t,t'\rangle$
stem from the following combinations:
\begin{align*}
&\eqref{CP1},~\eqref{CP3}\text{ on }\tr\lef\tr\rig \fa:&&
\langle \tr,~\tr\rangle,\\
&\eqref{CP1},~\eqref{CP4}\text{ on }x\lef(y\lef\tr\rig u)\rig v:&&
\langle x\lef y\rig v,~(x\lef y\rig v)\lef\tr\rig(x\lef u\rig v)\rangle,\\
&\eqref{CP2},~\eqref{CP3}\text{ on }\tr\lef\fa\rig \fa:&&
\langle \fa,~\fa\rangle,\\
&\eqref{CP2},~\eqref{CP4}\text{ on }x\lef(y\lef\fa\rig u)\rig v:
&&
\langle x\lef u\rig v,~(x\lef y\rig v)\lef\fa\rig(x\lef u\rig v)\rangle,\\
&\eqref{CP3},~\eqref{CP4}\text{ on }x\lef(\tr\lef z\rig \fa)\rig v:
&&
\langle x\lef z\rig v,~(x\lef \tr\rig v)\lef z\rig(x\lef \fa\rig v)\rangle,\\
&\eqref{CP3},~\eqref{CP4}\text{ on }\tr\lef(y\lef z\rig u)\rig \fa:
&&
\langle y\lef z\rig u,~(\tr\lef y\rig\fa)\lef z\rig(\tr\lef u\rig\fa)\rangle,
\end{align*}
and $\eqref{CP4},~\eqref{CP4}\text{ on }x\lef(w\lef(y\lef z\rig u)\rig r)\rig v$:
\[
\langle 
(x\lef w\rig v)\lef (y\lef z\rig u)\rig (x\lef r\rig v),
~x\lef((w\lef y\rig r)\lef z\rig(w\lef u\rig r))\rig v\rangle
\]
with
common reduct
\[((x\lef w \rig v)\lef y \rig (x\lef r\rig v ))\lef z\rig{}
((x\lef w \rig v)\lef u \rig (x\lef r\rig v )).
\]
Hence we have a ground-complete TRS, and a closed term $t$ is a
normal form if, and only if, $t\in A\cup\{\tr,\fa\}$, or $t$ satisfies
the following property:
\[
\text{ If $t_1\lef t_2\rig t_3$ is a subterm of $t$, then $t_2 
\in A$ and it is not the case that $t_1\equiv \tr$ and $t_3\equiv\fa$.}
\]
However, the
normal forms resulting from this TRS
are not particularly
suitable for systematic reasoning, and we introduce
another class of closed terms for this purpose.

\begin{definition}
\label{def:bf}
A term $t\in\NT_{\SigmaStateless}$ is a 
\textbf{basic form} if for $a\in A$,
\[t::=\tr\mid\fa\mid t\lef a \rig t. \]
\end{definition}

\begin{lemma}
\label{lem:nf}
For each closed term $t\in\NT_{\SigmaStateless}$
there exists a unique basic form $t'$ with $\CP\vdash t=t'$.
\end{lemma}

\begin{proof}
Let $t''$ be the unique normal form of $t$. 
Replace in $t''$ each subterm that is a 
single atom $a$ by $\tr\lef a\rig\fa$. This
results in a unique basic form $t'$ and clearly 
$\CP\vdash t=t'$.
\end{proof}

Let $S$ be a non-empty sort of states with constant $c$.
We extend the signature $\SigmaStateless$ to
\[\SigmaState=\SigmaStateless\cup
\{c:S,~.\lef.\rig.:S\times C\times S\rightarrow S\},\]
where \emph{sce} stands for ``states and
conditional expressions''.

\begin{definition}
\label{def:2PA}
A $\SigmaState$-algebra is a \textbf{two-sorted proposition algebra} 
if its $\SigmaStateless$-reduct is a proposition
algebra, and if it satisfies
the following axioms where $x,y,z$ range over conditional
expressions and $s,s'$ range over states:
\begin{align}
\label{2S1}\tag{TS1}
s\lef\tr\rig s'&=s,\\
\label{2S2}\tag{TS2}
s\lef\fa\rig s'&=s',\\
\label{2S3}\tag{TS3}
x\ne\tr\wedge x\ne\fa&\rightarrow 
s\lef x\rig s'=c.	
\end{align}
\end{definition}

Later on (after the next definition) we comment on these axioms.

\begin{proposition}
\label{prop:2s}
If $\CP\vdash t=t'$, then
$s\lef t\rig s'=s\lef t'\rig s'$
holds in each two-sorted proposition algebra.
\end{proposition}

So, the state set of a two-sorted proposition 
algebra can be seen as one that is
equipped with an if-then else construct and conditions that
stem from \CP.  
We extend the signature $\SigmaState$ to
\[\SigmaStateful=\SigmaState\cup
\{\apply:C\times S\rightarrow S,
~\reply:C\times S\rightarrow C\},\]
where \emph{spa} stands for ``stateful proposition algebra''
(see below).
The operator $\apply$ is called ``apply'' and the
operator $\reply$ is called ``reply'' and we further
assume that these operators bind stronger than
conditional composition.  The apply
and reply operator are taken from~\cite{BM09}. 

\begin{definition}
\label{def:SPA}
A $\SigmaStateful$-algebra is a \textbf{stateful proposition algebra}, 
SPA for short, 
if its reduct to $\SigmaState$ is a two-sorted proposition 
algebra, and
if it satisfies
the following axioms where $x,y,z$ range over conditional
expressions and $s$ ranges over states:
\begin{align}
\label{SPA1}
\tag{SPA1}
\tr\reply s&=\tr,\\
\label{SPA2}
\tag{SPA2}
\fa\reply s&=\fa,\\
\label{SPA3}
\tag{SPA3}
(x\lef y\rig z)\reply s&=
x\reply (y\apply s)\lef y\reply s\rig
z\reply(y\apply s),\\
\label{SPA4}
\tag{SPA4}
\tr\apply s&=s,\\
\label{SPA5}
\tag{SPA5}
\fa\apply s&=s,\\
\label{SPA6}
\tag{SPA6}
(x\lef y\rig z)\apply s&=
x\apply (y\apply s)\lef y\reply s\rig
z\apply (y\apply s),\\
\label{SPA7}
\tag{SPA7}
x\reply s=\tr&\vee x\reply s =\fa,\\
\label{SPA8}
\tag{SPA8}
\forall s(x\reply s=y\reply s\wedge x\apply s=y\apply s)&\rightarrow
x=y.
\end{align}
We refer to \eqref{SPA7} as \textbf{two-valuedness} and
we write $\CPS$ (for \textup{C}P and \textup{T}S and \textup{S}PA) 
for the set that contains all fifteen axioms involved.
\end{definition}

In a stateful proposition algebra $\NS$
with domain $C'$ of conditional expressions
and domain $S'$ of states, a conditional
expression $t$ can be associated with 
a `valuation function'
$t\reply:S'\rightarrow\{\tr,\fa\}$ (the evaluation 
of $t$ in some initial state) and
a `state transformer' 
$t\apply{}:S'\rightarrow S'$.

We note that the axioms of a SPA are consistent with those of
a two-sorted proposition algebra, and that
the special instances 
\[s\lef\tr\rig s=s\quad\text{and}\quad
s\lef\fa\rig s=s\]
of axioms~\eqref{2S1} and \eqref{2S2} are derivable:
first note that $s=\tr\apply(\tr\apply s)$ by axiom~\eqref{SPA4}, 
$\tr=\tr\reply s$ by \eqref{SPA1}, and
$\tr\lef\tr\rig\tr=\tr$
by CP-axiom~\eqref{CP1}, and thus
\begin{align*}
\CPS\vdash s\lef \tr \rig s
&=\tr\apply(\tr\apply s)\lef \tr\reply s\rig \tr\apply(\tr\apply s)\\
&=(\tr\lef \tr\rig \tr)\apply s\\
&=\tr\apply s\\
&= s.
\end{align*}
(Note that more derivable \CP-identities can be used to prove 
this fact, e.g., $\tr\lef\tr\rig\fa=\tr$.)
In a similar way one can derive
$s\lef\fa\rig s=s$. 

\begin{definition}
\label{def:hma}
A \textbf{Hoare-McCarthy algebra}, HMA for short, is the 
$\SigmaStateless$-reduct of
a stateful proposition algebra.
\end{definition}

For each HMA \BA\ we have by definition
$\BA\models \CP$. In Theorem~\ref{thm:1} below
we prove the existence of an HMA that characterizes \CP\
in the sense that a closed equation is valid only if it is
derivable from \CP. 

Recall $\NT_{\SigmaHMA}$ is the set of closed terms
over $\SigmaHMA$. We define 
\emph{structural congruence}, notation 
\[=_{sc}\]
on $\NT_{\SigmaHMA}$ as the congruence generated by
$\CP$.

\begin{theorem}
\label{thm:1}
An HMA that characterizes \CP\ exists: there is an
HMA $\BA^{sc}$ such that for all 
$t,t'\in\NT_{\SigmaHMA}$,
$\CP\vdash t=t'\iff \BA^{sc}\models t=t'$.
\end{theorem}

\begin{proof}
We construct the $\SigmaStateful$-algebra $\NS^{sc}$
with
$C'=\NT_{\SigmaHMA}/_{=_{sc}}$ as its set of
conditional expressions and the function space
\[S'=\{\tr,\fa\}^{A^+}\]
as its set of states. For 
each state $f$
and atom $a\in A$ define 
$a\reply f=f(a)$ and $a\apply f$ as the function defined 
for $\sigma\in A^+$ by 
\[(a\apply f)(\sigma)=f(a\sigma).\]
The state constant $c$ is given an arbitrary interpretation, and
the axioms~\eqref{2S1}--\eqref{2S3}
define $.\lef .\rig.:S'\times C'\times S'$ in $\NS^{sc}$.
The axioms~\eqref{SPA1}--\eqref{SPA6}
fully determine
the functions $\reply$ and $\apply$, and this is well-defined:
if $t=_{sc}t'$ then for all $f$, $t\reply f=t'\reply f$ and
$t\apply f=t'\apply f$ (this follows by inspection of the \CP\
axioms). The axiom~\eqref{SPA7} holds by construction of $S'$.
In order to prove that $\NS^{sc}$ is a SPA 
it remains to be shown that 
axiom~\eqref{SPA8} holds, i.e.,
for all $t,t'\in\NT_{\SigmaHMA}$,
\[\forall f(t\reply f=t'\reply f\wedge
t\apply f=t'\apply f)\rightarrow t=_{sc}t'.\]
This follows by contraposition.
By Lemma~\ref{lem:nf} we may assume that $t$ and $t'$ are
basic forms, and we apply induction on the
complexity of $t$.

\begin{enumerate}
\item
If $t\equiv \tr$, then $t'\equiv\fa$ yields $t\reply f
\ne t'\reply f$ for any $f$,  and if 
$t'\equiv t_1\lef a\rig t_2$ then consider 
$f$ with $f(a)=\tr$ and 
$f(a\sigma)=\fa$ for
$\sigma\in A^+$. We find $t\apply f=f$ and
$t'\apply f\ne f$ because $(t'\apply f)(a)=(t_1\apply f)(a\sigma)=\fa$.

\item
If $t\equiv\fa$ a similar argument applies.

\item
If $t\equiv t_1\lef a\rig t_2$, then the case 
$t'\in\{\tr,\fa\}$
can be dealt with as above. 

If $t'\equiv t_3\lef a\rig t_4$
then assume $t_1\lef a\rig t_2\neq_{sc}t_3\lef a\rig t_4$ because
$t_1\neq_{sc}t_3$. By induction there exists $f$ with
$t_1\apply f\neq t_3\apply f$ or
$t_1\reply f\neq t_3\reply f$.
Take some $g$ such that $a\apply g=f$ and $a\reply g=\tr$,
then $g$ distinguishes
$t_1\lef a\rig t_2$ and $t_3\lef a\rig t_4$. 
If $t_1=_{sc}t_3$, then a similar argument applies
for $t_2\neq_{sc}t_4$.

If $t'\equiv t_3\lef b\rig t_4$ with $a$ and $b$ different,
then 
$(t_1\lef a\rig t_2)\apply f\ne (t_3\lef b\rig t_4)\apply f$ 
for $f$ defined
by $f(a)=f(a\sigma)=\tr$ and $f(b)=f(b\sigma)=\fa$ because
$((t_1\lef a\rig t_2)\apply f)(a)=
(t_1\apply (a\apply f))(a)=
f(a\rho a)=\tr$, and
$((t_3\lef b\rig t_4)\apply f)(a)=(t_4\apply (b\apply f))(a)=
f(b\rho' a)=\fa$ (where
$\rho,\rho'$ possibly equal $\epsilon$).

\end{enumerate}
So $\NS^{sc}$ is a SPA. Define the HMA $\BA^{sc}$ 
as the $\SigmaHMA$-reduct of $\NS^{sc}$.  The
validity of axiom~\eqref{SPA8} proves $\Longleftarrow$
as stated in the theorem (the implication $\Longrightarrow$
holds by definition of a SPA).
\end{proof}

Observe that $\BA^{sc}\cong I(\SigmaHMA,\CP)$. By the
proof of the above theorem we find for all 
$t,t'\in\NT_{\SigmaHMA}$,
\begin{equation}
\label{eq:hma}
\CP\vdash t=t'\iff\NS^{sc}\models t=t'.
\end{equation}
We have the following (trivial) corollary on the 
quasivariety of SPAs, the first of a 
number of corollaries in which certain quasivarieties
of SPAs are characterized.

\begin{corollary}
\label{cor:0}
Let $\NC_{\fr}$ be the class of all SPAs. Then for all
$t,t'\in\NT_{\SigmaHMA}$,
\[\NC_{\fr}\models t=t'\iff \CP\vdash t=t'.\]
\end{corollary}

\begin{proof}
By the facts that $\NS^{sc}\in\NC_{\fr}$ and 
that each SPA satisfies \CP\ by definition.
\end{proof}

\section{Not all proposition algebras are HMAs}
\label{sec:2}
In this section we show that not all proposition algebras are HMAs.
Then we formulate a sufficient condition under which 
a proposition algebra is an HMA.

If we add to $\CP$ the equation 
\[
\tr\lef x\rig \tr=\tr, 
\]
a non-trivial initial algebra 
$I(\SigmaHMA,\CP+\langle\, \tr\lef x\rig \tr=\tr\,\rangle)$
exists:
again we can define normal forms by directing all axioms
from left
to right; this yields a strongly normalizing TRS
by the weight function $w$ defined in the previous section.
It is also not difficult to see that
this TRS is weakly confluent, the 
critical pairs not dealt with before arise from the
following combinations:
\begin{align*}
&\eqref{CP1}\text{ and }\langle\, \tr\lef x\rig \tr=\tr\,\rangle
\text{ on }\tr\lef\tr \rig \tr:&&
\langle \tr,~\tr\rangle,\\
&\eqref{CP2}\text{ and }\langle\, \tr\lef x\rig \tr=\tr\,\rangle
\text{ on }\tr\lef\fa \rig \tr:&&
\langle \tr,~\tr\rangle,\\
&\eqref{CP4}\text{ and }\langle\, \tr\lef x\rig \tr=\tr\,\rangle
\text{ on }\tr\lef(y\lef z\rig u)\rig \tr:&&
\langle 
(\tr\lef y\rig \tr)\lef z\rig (\tr\lef u\rig \tr),~\tr\rangle,\\
&\eqref{CP4}\text{ and }\langle\, \tr\lef x\rig \tr=\tr\,\rangle
\text{ on }x\lef(\tr\lef z\rig \tr)\rig u:
\end{align*}
\[\qquad
\langle 
(x\lef \tr\rig u)\lef z\rig(x\lef \tr\rig u),~x\lef \tr\rig u\rangle.\]

It is easily seen that all these pairs have a common reduct,
hence, also this TRS is ground-complete. Furthermore,
observe that both $\tr\lef a\rig b$ and 
$\tr\lef b\rig a$
are normal forms.
Note the following consequence in 
$\CP+\langle\, \tr\lef x\rig \tr=\tr\,\rangle$:
\begin{align}
\nonumber
x&=x\lef(\tr\lef y\rig\tr)\rig x\\
\nonumber
&=(x\lef\tr\rig x)\lef y\rig(x\lef\tr\rig x)\\
\label{X}
&=x\lef y\rig x.
\end{align}

Now consider the conditional equation
\begin{equation}
\label{eq:oh}
((\tr\lef x\rig \tr=\tr)\wedge(\tr\lef y\rig \tr=\tr))
\rightarrow \tr\lef x\rig y=\tr\lef y\rig x.
\end{equation}

\begin{lemma}
\label{lem:oh}
Each HMA satisfies the conditional equation \eqref{eq:oh}.
\end{lemma}

\begin{proof}
Let HMA \BA\ be the $\SigmaHMA$-reduct of some stateful
proposition algebra
$\NS$ with domains $C'$ and $S'$ and assume
$\BA,\sigma\models (\tr\lef x\rig \tr=\tr)\wedge (\tr\lef y\rig \tr=\tr)$
for some assignment $\sigma$. 
Writing $\sigma(x)=t$ and $\sigma(y)=t'$, it follows
that $\forall s\in S'(t\apply s=s=t'\apply s)$:
\begin{align*}
(\tr\lef t\rig \tr)\apply s
&=\tr\apply t\apply s\lef
t\reply s\rig
\tr\apply t\apply s\\
&=t\apply s\lef
t\reply s\rig t\apply s.
\end{align*}
By assumption $(\tr\lef t\rig \tr)\apply s=\tr\apply s$ and 
by $\tr\apply s=s$ we find by 
axiom \eqref{SPA7}
and axiom~\eqref{2S3}
that $t\apply s=s$. In a similar way it follows
that $t'\apply s=s$. 

Then for all $s\in S'$,
\begin{align*}
(\tr\lef t\rig t')\reply s
&=\tr\reply (t\apply s)\lef
t\reply s\rig
t'\reply(t\apply s)\\
&=\tr\lef
t\reply s\rig t'\reply s, 
\end{align*}
and by symmetry, 
$(\tr\lef t'\rig t)\reply s=\tr\lef t'\reply s\rig t\reply s$.
Now $(\tr\lef t\rig t')\reply s=(\tr\lef t'\rig t)\reply s$
follows by case distinction, using axioms~\eqref{SPA7}, 
\eqref{CP1} and \eqref{CP2}. Furthermore, by~\eqref{SPA7} 
and~\eqref{2S1}, 
\begin{align*}
(\tr\lef t\rig t')\apply s
&=\tr\apply t\apply s\lef
t\reply s\rig
t'\apply t\apply s\\
&=s\lef
t\reply s\rig s\\
&=s,
\end{align*}
and in a similar way it follows that 
$(\tr\lef t'\rig t)\apply s= s$, thus 
$(\tr\lef t\rig t')\apply s=(\tr\lef t'\rig t)\apply s$.
By~\eqref{SPA8}, $\tr\lef t\rig t'=\tr\lef t'\rig t$ and
thus
 $\BA,\sigma\models  \tr\lef x\rig y=\tr\lef y\rig x$, 
as was to be proved.
\end{proof}

In a setting with two different atoms, not each 
proposition algebra is an HMA.

\begin{theorem}
\label{thm:oh}
For $|A|>1$ there exist proposition algebra's that are no HMAs.
\end{theorem}

\begin{proof}
Consider the initial algebra 
$I(\SigmaHMA,\CP+\langle\, \tr\lef x\rig \tr=\tr\,\rangle)$.
Clearly this algebra satisfies $\tr\lef a\rig\tr=\tr
=\tr\lef b\rig\tr$, and therewith an instance of the premise of 
conditional equation~\eqref{eq:oh}, but not its conclusion 
$\tr\lef a\rig b=\tr\lef b\rig a$ because these terms are different 
normal forms. 
By Lemma~\ref{lem:oh}, 
$I(\SigmaHMA,\CP+\langle\, \tr\lef x\rig \tr=\tr\,\rangle)$ is not an HMA.
\end{proof}

Let $\HMA_A$ be the class of $\SigmaHMA$-algebra's
that are HMAs.
The diagram of $\BA\in\HMA_A$, notation
$\Delta_\BA$, is defined by
\[\Delta_\BA=\{t=t'\mid t,t'\in\NT_{\SigmaHMA},
~\BA\models t=t'\}\cup
\{t\neq t'\mid t,t'\in\NT_{\SigmaHMA},~\BA\models t\neq t'\}.\]
Let $\CCE(\BA)$ be the closed conditional equational 
theory of \BA\ 
and let $\CCE(\HMA_A)$ be the set of closed conditional equations true in
all HMAs, thus 
\[\CCE(\HMA_A)=\bigcap_{\BA\in\HMA_A}\CCE(\BA).\]

\begin{theorem}
Let $\BA$ be some minimal $\SigmaHMA$-algebra. If 
$\BA\models \CCE(\HMA_A)$ then $\BA\in\HMA_A$.
\end{theorem}

\begin{proof}
Using compactness we prove that $\Delta_\BA\cup\CPS$ 
is consistent. Consider finite subsets 
$D$ and $D'$ of the positive
respectively negative part of $\Delta_\BA$.

If $D'=\emptyset$, then extend $\BA$ to
a two-sorted model $\NS^s$ by adding a state set $S'=\{s\}$
and defining the function $s\lef t\rig s$ by the
axioms~\eqref{2S1}--\eqref{2S3}
(of course, the interpretation of the state constant $c$
is $s$).
Furthermore, define in $\NS^s$ the functions
$\reply$ and $\apply$ by
$a\reply s=\tr$ and 
$a\apply s=s$ for all $a\in A$,
and the other cases by axioms \eqref{SPA1}--\eqref{SPA6}
(so $t\apply s=s$ for all $t$).
Finally, if for closed terms $t$ and $t'$,
$t\reply s=t'\reply s$,
extend $D$ with $t=t'$.
Now observe that
the axioms of \CP\ are valid in $\NS^s$
because $\BA\models \CCE(\HMA_A)$.
Furthermore, axiom~\eqref{SPA7} is trivially valid.
Axiom~\eqref{SPA8} is valid by construction, 
so $\NS^s\models\CPS\cup D$.

If $D'\neq\emptyset$, then let $e'$ 
be such that $\neg e'\in D'$.
Let $E=\bigwedge_{e\in D}e$ and write $\neg D'$ for
the set of equations whose negation is in $D'$, so 
$e\in \neg D'$ if and only if $\neg e\in D'$.
Then $E\rightarrow e'\not\in\CCE(\HMA_A)$
because $\BA\not\models E\rightarrow e'$.
Thus there exists a model $\NS_{e'}$ of
$\CPS\cup E\cup\{\neg e'\}$. We can consider a 
disjoint union $\NS^*$ of all $\NS_{e'}$ for
$e'\in D'$, where we forget all $c$'s (the state
constant that guarantees that $S$ is a non-empty sort).
Here the state sets are taken disjoint and 
for $D'=\{\neg e_1',...,\neg e_n'\}$,
$S_{\NS^*}=S_{\NS_{e_1'}}\cup...\cup S_{\NS_{e_n'}}$.
The disjoint union then found is again a model of
$\CPS\cup E$ and it satisfies $\neg e'$ for each
$e'\in \neg D'$. Finally, $c$ is given an arbitrary
interpretation. We find that
$\NS^*\models \CPS\cup E\cup\{\neg e'\mid e'\in \neg D'\}$.

By compactness this proves the consistency of 
$\Delta_{\BA} \cup\CPS $. 
Let $\NS$ be a $\SigmaStateful$-algebra with
$\NS\models \Delta_\BA\cup\CPS$.
Then the minimal subalgebra $\NS'$  of \NS\ is
a model of \CPS\ and its reduct to $\SigmaHMA$
satisfies $\Delta_\BA$. 
So $\NS'\restriction \SigmaHMA\cong\BA$,
whence $\BA\in\HMA_A$.
\end{proof}

\section{Repetition-proof congruence}
\label{sec:rp}
In this section we consider \emph{repetition-proof congruence}
defined by the axioms of $\CP$ and 
these axiom schemes ($a\in A$):
\begin{align*}
\label{CPrp1}\tag{CPrp1} \qquad
(x\lef a\rig y)\lef a\rig z&=
(x\lef a\rig x)\lef a\rig z,\\
\label{CPrp2}\tag{CPrp2} \qquad
x\lef a\rig (y\lef a\rig z)&= x \lef a\rig (z\lef a\rig z).
\end{align*}
Typically,
the valuation of successive equal atoms yields the same reply.

We write $\CP_{\rp}$ for this set of axioms. 
Let \emph{repetition-proof congruence}, notation
$=_{\rp}$, be the congruence on $\NT_{\SigmaHMA}$ generated 
by the axioms of $\CP_{\rp}$.  
\begin{definition}
\label{def:rpbf}
A term $t\in\NT_{\SigmaStateless}$ is an 
\textbf{rp-basic form} if for $a\in A$,
\[t::=\tr\mid\fa\mid t_1\lef a \rig t_2 \]
and $t_i$ ($i=1,2$) is an \rp-basic form with
the restriction that the central
condition (if present) is either
different from $a$, or 
$t_i\equiv t_i'\lef a\rig t_i'$ with $t_i'$ 
 an \rp-basic form.
\end{definition}

\begin{lemma}
\label{lem:rp}
For each $t\in\NT_{\SigmaHMA}$ there exists an \rp-basic form
$t'$ with $\CP_{\rp}\vdash t=t'$.
\end{lemma}

\begin{proof}
First, we prove that the conditional composition $t_1\lef t_2\rig t_3$ of three
\rp-basic terms  can be proved equal to an \rp-basic term by
structural induction on $t_2$. If $t_2\in\{\tr,\fa\}$ this is 
trivial, and otherwise we find by induction \rp-basic forms
$t_4$ and $t_5$ with
\begin{align*}
t_1\lef(t\lef a \rig t')\rig t_3
&=(t_1\lef t\rig t_3)\lef a \rig(t_1\lef t'\rig t_3)\\
&=t_4\lef a\rig t_5.
\end{align*}
If $t_4\equiv t_6\lef a \rig t_7$ then apply axiom 
\eqref{CPrp1} on $t_4$, thus obtaining
$t_6\lef a\rig t_6$, and if $t_5\equiv t_8\lef a \rig t_9$, replace
it by $t_9\lef a\rig t_9$ (axiom \eqref{CPrp2}). 
Clearly, the resulting  term is an \rp-basic form.

With the above result, the lemma's statement follows 
easily by structural induction. 
\end{proof}

\begin{theorem}
\label{thm:rp}
For $|A|>1$,
an HMA that characterizes $\CP_{\rp}$
exists: there is an
HMA $\BA^{\rp}$ such that for all 
$t,t'\in\NT_{\SigmaHMA}$,
$\CP_{\rp}\vdash t=t'\iff \BA^{\rp}\models t=t'$.
\end{theorem}

\begin{proof}
Define the function space 
\[\RP\subset\{\tr,\fa\}^{A^+}\]
by $f\in\RP$ if for all $a\in A$ and $\sigma\in A^*$,
$f(\sigma aa)=f(\sigma a)$.
Construct the $\SigmaStateful$-algebra $\NS^{\rp}$
with $\NT_{\SigmaHMA}/_{=_{\rp}}$ as its set of
conditional expressions and \RP\
as its set of states. For each state $f$
and atom $a\in A$ define 
$a\reply f=f(a)$
and $a\apply f$ by
\[(a\apply f)(\sigma)=f(a\sigma).\]
Clearly, if $f\in\RP$ then $a\apply f\in\RP$.

Similar as in the proof of Theorem~\ref{thm:1},
the state constant $c$ is given an arbitrary interpretation, and
the axioms~\eqref{2S1}--\eqref{2S3}
define the function $s\lef f\rig s'$ in $\NS^{\rp}$.
The axioms~\eqref{SPA1}--\eqref{SPA6}
fully determine
the functions $\reply$ and $\apply$, and this is well-defined:
if $t=_{\rp}t'$ then for all $f$, $t\reply f=t'\reply f$
and $t\apply f=t'\apply f$ follow by inspection of the 
$\CP_{\rp}$ axioms. We show soundness of the axiom
scheme~\eqref{CPrp1}: For all $f\in\RP$,
$a\reply (a\apply f)=a\reply f$, and thus if 
$a\reply f=\tr$,
\[(t_1\lef a\rig t_2)\reply(a\apply f)=(t_1\lef a\rig t_1)\reply(a\apply f).\] 
We derive
\begin{align*}
((t_1\lef a\rig t_2)\lef a\rig t)\reply f
&=(t_1\lef a\rig t_2)\reply(a\apply f)\lef
a\reply f\rig t\reply(a\apply f)\\
&=(t_1\lef a\rig t_1)\reply(a\apply f)\lef
a\reply f\rig t\reply(a\apply f)\\
&=((t_1\lef a\rig t_1)\lef a\rig t)\reply f,
\end{align*}
and
\begin{align*}
((t_1\lef a\rig t_2)\lef a\rig t)\apply f
&=(t_1\lef a\rig t_2)\apply(a\apply f)\lef
a\reply f\rig t\apply(a\apply f)\\
&=(t_1\lef a\rig t_1)\apply(a\apply f)\lef
a\reply f\rig t\apply(a\apply f)\\
&=((t_1\lef a\rig t_1)\lef a\rig t)\apply f.
\end{align*}
The soundness of \eqref{CPrp2} follows in a similar
way.
The axiom~\eqref{SPA7} holds by construction of \RP.
In order to prove that $\NS^{\rp}$ is a SPA 
it remains to be shown that 
axiom~\eqref{SPA8} holds, i.e.,
for all $t,t'\in\NT_{\SigmaHMA}$,
\[\forall f(t\reply f=t'\reply f\wedge
t\apply f=t'\apply f)\rightarrow t=_{\rp}t'.\]
This follows by contraposition in the same
way as in the proof of Theorem~\ref{thm:1}.
However, the restriction to \RP\ imposes some subtle
constraints,
so we give a full proof. We may assume that
both $t$ and $t'$ are \rp-basic forms. We 
apply induction on the
complexity of $t$. 
Let $a,b\in A$ with $a\ne b$.
\begin{enumerate}
\item
If $t\equiv \tr$, then if $t'\equiv\fa$ 
it follows that $t\reply f
\ne t'\reply f$ for any $f\in \RP$,  and if 
$t'\equiv t_1\lef a\rig t_2$ then consider 
some $f\in\RP$ with $f(b)=\tr$ and 
$f(a)=f(a\sigma)=\fa$ for
$\sigma\in A^+$. We find $(t\apply f)(b)=f(b)=\tr$ 
and
$(t'\apply f)(b)=(t_1\apply f)(a\sigma b)=\fa$ 
(where $\sigma$ possibly equals $\epsilon$), 
so $t'\apply f\ne t\apply f$.

\item
If $t\equiv\fa$ a similar argument applies.

\item
If $t\equiv t_1\lef a\rig t_2$, then the case 
$t'\in\{\tr,\fa\}$
can be dealt with as above. 

If $t'\equiv t_3\lef a\rig t_4$
then assume $t\neq_{\rp}t'$ because
$t_1\neq_{\rp}t_3$. By induction there exists $f\in\RP$ with
the distinguishing property
$t_1\apply f\neq t_3\apply f$ or
$t_1\reply f\neq t_3\reply f$. 
\begin{itemize}
\item
If none of $t_1$ and $t_3$ has $a$ as its central condition, 
there exists $g\in RP$ with $a\apply g=f$ and 
$a\reply g=\tr$, and such a function
$g$ distinguishes $t$ and $t'$.
\item
If at least one of $t_1$ and $t_3$ has $a$ as its central
condition,
then this $a$ and all successive $a$'s occur
in subterms of the form $t''\lef a\rig t''$ because $t$ and 
$t'$ are \rp-basic forms. Hence, we may assume that
$t_1$ and $t_3$ can be distinguished by $f'\in\RP$ with 
$f'(a)=\tr$ and $f'$ otherwise defined as $f$ (so,
$f$ and $f'$ differ at most on initial $a$-sequences).
We find that $f'$ distinguishes $t$ and $t'$.
\end{itemize}
If $t_1=_{\rp}t_3$, then a similar argument applies
for $t_2\neq_{\rp}t_4$.

If $t'\equiv t_3\lef b\rig t_4$ with $a$ and $b$ different,
then 
$(t_1\lef a\rig t_2)\apply f\ne (t_3\lef b\rig t_4)\apply f$ 
for $f$ defined
by $f(a)=f(a\sigma)=\tr$ and $f(b)=f(b\sigma)=\fa$ because
$((t_1\lef a\rig t_2)\apply f)(a)=
(t_1\apply (a\apply f))(a)=
f(a\rho a)=\tr$, and
$((t_3\lef b\rig t_4)\apply f)(a)=(t_4\apply (b\apply f))(a)=
f(b\rho' a)=\fa$ (where
$\rho,\rho'$ possibly equal $\epsilon$).

\end{enumerate}
So $\NS^{\rp}$ is a SPA. 
Define the HMA $\BA^{\rp}$ as the $\SigmaHMA$-reduct 
of $\NS^{\rp}$. The above argument on the soundness of
the axiom schemes~\eqref{CPrp1} and \eqref{CPrp2}
proves $\Longrightarrow$ as stated in the theorem,
and the
validity of axiom~\eqref{SPA8} proves $\Longleftarrow$.
We finally note that $\BA^{\rp}\cong I(\SigmaHMA,\CP_{\rp})$.
\end{proof}

In the proof above we defined the SPA $\NS^{\rp}$ and 
we found that if $|A|>1$, then 
for all $t,t'\in\NT_{\SigmaHMA}$,
\begin{equation}
\label{eq:ccrp}
\CP_{\rp}\vdash t=t' \iff \NS^{\rp}\models t=t'.
\end{equation}
If $A=\{a\}$ then $\NS^{\rp}$ 
has only two states, say $f$ and $g$ with
$f(a^{n+1})=\tr$ and $g(a^{n+1})=\fa$ and it
easily follows that 
\[\BA^{\rp}\models \tr\lef a\rig\tr=\tr,\]
so 
$\BA^{\rp}\not\cong I(\SigmaHMA,\CP_{\rp})$ in this case.
The following corollary is related to Theorem~\ref{thm:rp} 
and characterizes repetition-proof congruence
in terms of a quasivariety of SPAs that satisfy an extra condition. 

\begin{corollary}
\label{cor:rp}
Let $|A|>1$.
Let $\NC_{\rp}$ be the class of SPAs that satisfy
for all $a\in A$ and $s\in S$,
\[a\reply(a\apply s)= a\reply s.\]
Then for all
$t,t'\in\NT_{\SigmaHMA}$,\[\NC_{\rp}\models t=t'\iff \CP_{\rp}\vdash t=t'.\]
\end{corollary}

\begin{proof}
By its definition, $\NS^{\rp}\in \NC_{\rp}$, 
which by \eqref{eq:ccrp} implies $\Longrightarrow$.
For the converse, it is sufficient to show that
the axioms
\eqref{CPrp1} and \eqref{CPrp2} hold in each SPA
that is in $\NC_{\rp}$.
Let such $\NS$ be given.
Consider \eqref{CPrp1}: if for some
interpretation of $s$ in \NS, $a\reply s=\fa$ 
there is nothing to prove, and if $a\reply s=\tr$, then
$a\reply(a\apply s)=\tr$ and hence
\begin{align*}
((t_1\lef a\rig t_2)\lef a\rig t)\reply s&=
t_1\reply (a\apply (a\apply s))\\
&= ((t_1\lef a\rig t_1)\lef a\rig t)\reply s,
\end{align*}
and
\begin{align*}
((t_1\lef a\rig t_2)\lef a\rig t)\apply s&=
t_1\apply (a\apply (a\apply s))\\
&= ((t_1\lef a\rig t_1)\lef a\rig t)\apply s.
\end{align*}
The soundness of axiom \eqref{CPrp2} can be proved
in the same way.
\end{proof}

\section{Contractive congruence}
\label{sec:cr}
In this section we consider \emph{contractive  
congruence}
defined by the axioms of $\CP$ and 
these axiom schemes ($a\in A$):
\begin{align*}
\label{CPcr1}\tag{CPcr1} \qquad
(x\lef a\rig y)\lef a\rig z&=
x\lef a\rig z,\\
\label{CPcr2}\tag{CPcr2} \qquad
x\lef a\rig (y\lef a\rig z)&= x \lef a\rig z.
\end{align*}
Typically, successive equal atoms are contracted.

We write $\CP_{\con}$ for this set of axioms. 
Let \emph{contractive congruence}, notation
$=_{\con}$, be the congruence on $\NT_{\SigmaHMA}$ generated 
by the axioms of $\CP_{\con}$.  

\begin{definition}
\label{def:crbf}
A term $t\in\NT_{\SigmaStateless}$ is a
\textbf{cr-basic form} if for $a\in A$,
\[t::=\tr\mid\fa\mid t_1\lef a \rig t_2 \]
and $t_i$ ($i=1,2$) is a \con-basic form with
the restriction that the
central
condition (if present) is
different from $a$.
\end{definition}

\begin{lemma}
For each $t\in\NT_{\SigmaHMA}$ there exists a \con-basic form
$t'$ with $\CP_{\con}\vdash t=t'$.
\end{lemma}

\begin{proof}
Similar to the proof of Lemma~\ref{lem:rp}.
\end{proof}

\begin{theorem}
\label{thm:cr}
For $|A|>1$ an HMA that characterizes
$\CP_\con$ exists, i.e. there is an
HMA $\BA^{\con}$ such that for all 
$t,t'\in\NT_{\SigmaHMA}$,
$\CP_{\con}\vdash t=t'\iff \BA^{\con}\models t=t'$.
\end{theorem}

\begin{proof}
Let $A^{\con}\subset A^+$ be the set of strings that
contain no consecutive occurrences of the same atom.
Construct the $\SigmaStateful$-algebra $\NS^{\con}$
with
$\NT_{\SigmaHMA}/_{=_{\con}}$ as its set of
conditional expressions and the function space
\[\{\tr,\fa\}^{A^{\con}}\]
as its set of states. For 
each state $f$ and atom $a\in A$ define 
$a\reply f=f(a)$ and $a\apply f$ by
\[(a\apply f)(\sigma)=
\begin{cases}
f(\sigma)&\text{if $\sigma=a$ or $\sigma=a\rho$},\\
f(a\sigma)&\text{otherwise}.
\end{cases}
\]
Clearly, $a\apply f\in\{\tr,\fa\}^{A^{\con}}$ if 
$f\in\{\tr,\fa\}^{A^{\con}}$.
Similar as in the proof of Theorem~\ref{thm:1},
the state constant $c$ is given an arbitrary interpretation, and
the axioms~\eqref{2S1}--\eqref{2S3}
define the function $s\lef f\rig s'$ in $\NS^{\con}$.
The axioms~\eqref{SPA1}--\eqref{SPA6}
fully determine
the functions $\reply$ and $\apply$, and this is well-defined:
if $t=_{\con}t'$ then for all $f$, $t\reply f=t'\reply f$
and
$t\apply f=t'\apply f$ follow by inspection of the 
$\CP_{\con}$ axioms. We show soundness of 
the axiom scheme~\eqref{CPcr1}:
first note that $a\reply (a\apply f)=a\reply f$ and 
$a\apply (a\apply f)=a\apply f$, and
derive
\begin{align*}
((t_1\lef a\rig t_2)\lef a\rig t)\reply f
&=(t_1\lef a\rig t_2)\reply(a\apply f)\lef
a\reply f\rig t\reply(a\apply f)\\
&=t_1\reply(a\apply(a\apply f))\lef
a\reply f\rig t\reply(a\apply f)\\
&=(t_1\lef a\rig t)\reply f,
\end{align*}
and
\begin{align*}
((t_1\lef a\rig t_2)\lef a\rig t)\apply f
&=(t_1\lef a\rig t_2)\apply(a\apply f)\lef
a\reply f\rig t\apply(a\apply f)\\
&=t_1\apply(a\apply(a\apply f))\lef
a\reply f\rig t\apply(a\apply f)\\
&=(t_1\lef a\rig t)\apply f.
\end{align*}
The soundness of \eqref{CPcr2} follows in a similar
way.
The axiom~\eqref{SPA7} holds by construction of \RP.
In order to prove that $\NS^{\con}$ is a SPA 
it remains to be shown that 
axiom~\eqref{SPA8} holds, i.e.,
for all $t,t'\in\NT_{\SigmaHMA}$,
\[\forall f(t\reply f=t'\reply f\wedge
t\apply f=t'\apply f)\rightarrow t=_{\con}t'.\]
This follows by contraposition. We may assume that
both $t$ and $t'$ are \con-basic forms, and we 
apply induction on the
complexity of $t$. 
Let $a,b\in A$ with $a\ne b$.
\begin{enumerate}
\item
If $t\equiv \tr$, then if $t'\equiv\fa$ 
it follows that $t\reply f
\ne t'\reply f$ for any $f$,  and if 
$t'\equiv t_1\lef a\rig t_2$ then consider 
some $f$ with $f(b)=\tr$ and 
$f(a)=f(a\sigma)=\fa$ for
$a\sigma\in A^\con$. 
We find $(t\apply f)(b)=f(b)=\tr$ and
$(t'\apply f)(b)=(t_1\apply f)(a\sigma b)=\fa$ 
(where $\sigma$ possibly equals $\epsilon$), 
so $t'\apply f\ne t\apply f$.

\item
If $t\equiv\fa$ a similar argument applies.

\item
If $t\equiv t_1\lef a\rig t_2$, then the case 
$t'\in\{\tr,\fa\}$
can be dealt with as above. 

If $t'\equiv t_3\lef a\rig t_4$ then assume
$t\ne_\con t'$ because $t_1\ne_\con t_3$.
Then $a$ is not a central condition in 
$t_1$ and $t_3$, and by induction there exists $f$ with
$t_1\apply f\neq t_3\apply f$ or
$t_1\reply f\neq t_3\reply f$.
Take some $g$ such that $a\apply g=f$ and $a\reply g=\tr$,
then $g$ distinguishes
$t_1\lef a\rig t_2$ and $t_3\lef a\rig t_4$. 
If $t_1=_{\con}t_3$, then a similar argument applies
for $t_2\neq_{\con}t_4$.

If $t'\equiv t_3\lef b\rig t_4$ 
then 
$(t_1\lef a\rig t_2)\apply f\ne (t_3\lef b\rig t_4)\apply f$ 
for $f$ defined
by $f(a)=f(a\sigma)=\tr$ and $f(b)=f(b\sigma)=\fa$ because
$((t_1\lef a\rig t_2)\apply f)(b)=(t_1\apply (a\apply f))(b)=
f(a\rho b)=\tr$ (where
$\rho$ possibly equals $\epsilon$), and
$((t_3\lef b\rig t_4)\apply f)(b)=(t_4\apply (b\apply f))(b)$
and this equals either
$f(b\rho' b)=\fa$ for
some $\rho'\in(A\setminus\{b\})^\con$, or $f(b)=\fa$.
\end{enumerate}
So $\NS^{\con}$ is a SPA. 
Define the HMA $\BA^{\con}$ as the $\SigmaHMA$-reduct 
of $\NS^{\con}$. The above argument on the soundness of
the axiom schemes~\eqref{CPcr1} and \eqref{CPcr2}
proves $\Longrightarrow$ as stated in the theorem,
and the
validity of axiom~\eqref{SPA8} proves $\Longleftarrow$.
Finally, we note that $\BA^{\con}\cong I(\SigmaHMA,\CP_{\con})$.
\end{proof}

In the proof above we defined
the SPA $\NS^{\con}$ and we found that if
$|A|>1$, then
for all $t,t'\in\NT_{\SigmaHMA}$,
\begin{equation}
\label{eq:ccr}
\CP_{\con}\vdash t=t' \iff \NS^{\con}\models t=t'.
\end{equation}
If $A=\{a\}$ then $A^{\con}=A$ and 
$\NS^{\con}$ as defined above has only two states,
say $f$ and $g$ with
$f(a)=\tr$ and $g(a)=\fa$. It easily follows that 
\[\BA^{\con}\models \tr\lef a\rig\tr=\tr,\]
so 
$\BA^{\con}\not\cong I(\SigmaHMA,\CP_{\con})$ if $A=\{a\}$.
The following corollary is related to Theorem~\ref{thm:cr} 
and characterizes contractive congruence in terms
of a quasivariety of SPAs that satisfy an extra condition.

\begin{corollary}
\label{cor:cr}
Let $|A|>1$. 
Let $\NC_{\con}$ be the class of SPAs that satisfy
for all $a\in A$ and $s\in S$,
\[a\reply (a\apply s)= a\reply s~\wedge~
a\apply (a\apply s)= a\apply s.\]
Then for all $t,t'\in\NT_{\SigmaHMA}$,
\[\NC_{\con}\models t=t'\iff \CP_{\con}\vdash t=t'.\]
\end{corollary}

\begin{proof}
By its definition, $\NS^{\con}\in \NC_{\con}$, 
which by \eqref{eq:ccr} implies $\Longrightarrow$.
For the converse, it is sufficient to show that
the axioms
\eqref{CPcr1} and \eqref{CPcr2} hold in any SPA 
that is in $\NC_{\con}$.
Let such $\NS$ be given.
Consider \eqref{CPcr1}: 
if for some
interpretation of $s$ in \NS, $a\reply s=\fa$ 
there is nothing to prove, and if $a\reply s=\tr$, then
$a\reply(a\apply s)=\tr$ and hence
\begin{align*}
((t_1\lef a\rig t_2)\lef a\rig t)\reply s&=
t_1\reply
(a\apply (a\apply s))\\
&=t_1\reply (a\apply s)\\
&=(t_1\lef a\rig t)\reply s,
\end{align*}
and
\begin{align*}
((t_1\lef a\rig t_2)\lef a\rig t)\apply s&=
t_1\apply
(a\apply (a\apply s))\\
&=t_1\apply (a\apply s)\\
&=(t_1\lef a\rig t)\apply s.
\end{align*}
The soundness of axiom \eqref{CPcr2} can be proved
in the same way.
\end{proof}

\section{Weakly memorizing congruence}
\label{sec:wmem}
In this section we consider \emph{weakly memorizing congruence}
defined by the axioms of $\CP_\con$ and 
these axiom schemes ($a,b\in A$):
\begin{align*}
\label{CPwm1}\tag{CPwm1} \qquad
((x\lef a\rig y)\lef b\rig z)\lef a\rig v&=
(x\lef b\rig z)\lef a\rig v,\\
\label{CPwm2}\tag{CPwm2} \qquad
x\lef a\rig (y\lef b\rig (z\lef a\rig v))
&= x \lef a\rig (y\lef b\rig v).
\end{align*}
Note that for $a=b$, these axioms follow from  $\CP_\con$.
We write $\CP_{\wmem}$ for this set of axioms.
Typically, if evaluation of a series of successive
atoms yields equal replies, contraction takes place.
This is also the case if there
is more than one ``intermediate'' atom, an example is
\begin{align*}
(((x\lef a\rig y)\lef b\rig z)\lef c\rig u)\lef a \rig v
&=((((x\lef a\rig y)\lef b\rig z)\lef a\rig w)\lef
 c\rig u)\lef a\rig v\\
&=(((x\lef b\rig z)\lef a\rig w)\lef c\rig u)\lef a\rig v\\
&=((x\lef b\rig z)\lef c\rig u)\lef a \rig v.
\end{align*}

Let \emph{weakly memorizing congruence}, 
notation $=_{\wmem}$, be the congruence on $\NT_{\SigmaHMA}$ generated 
by the axioms of $\CP_{\wmem}$.  
Again we define a special type of basic forms.

\begin{definition}
\label{def:wmembf}
Let $t$ be a basic form. 
Then $pos(t)$ is the set of atoms that occur as the central
condition of $t$, or at a left-hand (positive) position in $t$:
\[pos(\tr)=pos(\fa)=\emptyset \quad\text{and}\quad
pos(t\lef a\rig t')=\{a\}\cup pos(t),\]
and $neg(t)$ is the set of atoms that occur as the central
condition of $t$, or at a right-hand (negative) position in $t$:
\[neg(\tr)=neg(\fa)=\emptyset \quad\text{and}\quad
neg(t\lef a\rig t')=\{a\}\cup neg(t').\]
Term $t\in\NT_{\SigmaStateless}$ is a
\textbf{wmem-basic form} if for $a\in A$,
\[t::=\tr\mid\fa\mid t_1\lef a \rig t_2 \]
and $t_1$ and $t_2$ are \wmem-basic forms with
the restriction that $a\not\in pos(t_1)\cup neg(t_2)$.
\end{definition}

\begin{lemma}
\label{lem:wm}
For each $t\in\NT_{\SigmaHMA}$ there exists a \wmem-basic form
$t'$ with $\CP_{\wmem}\vdash t=t'$.
\end{lemma}
\begin{proof}
See~\cite{BP10}; this proof is repeated in
Appendix~\ref{app:Proofs}.
\end{proof}

In the following we prepare the ingredients for an
HMA that characterizes $=_\wmem$.
Recall $A^{\con}\subset A^+$ is the set of strings that
contain no consecutive occurrences of the same atom. Define 
``element-wise
left-concatenation with absorption''
$\lcon$ on 
$A\times A^\con\rightarrow A^\con$ by
\[a\lcon\sigma=\begin{cases}
a&\text{if $\sigma=a$},\\
a\lcon\rho&\text{if $\sigma=a\rho$},\\
a\sigma&\text{otherwise}.
\end{cases}
\]
Observe that for all $\sigma\in A^\con$, $a\lcon(a\lcon\sigma)
=a\lcon \sigma$.

\begin{definition}
\label{def:WM}
The function space 
$\WM\subset \{\tr,\fa\}^{A^\con}$ is defined
by the following restriction: 
$f\in\WM$ if for all $a\in A$ and $b\in A\setminus\{a\}$,
and all $\rho\in A^*$ that satisfy $\rho a\in A^\con$,
\[f(\rho a b)=f(\rho a)\Longrightarrow
\begin{cases}
f(\rho aba)=f(\rho a),
\text{ and}\\
f(\rho aba\lcon\sigma)=f(\rho ab\lcon\sigma)&\text{for all
$\sigma\in A^\con$}.
\end{cases}
\]
\end{definition}
For example, if $f\in\WM$ and $b\sigma\in A^\con$, then
\begin{equation}
\label{eq:vbwm}
f(a)=f(ab)~\Longrightarrow~ (f(abab)=f(ab)\quad\text{and}
\quad
f(abab\sigma)=f(ab\sigma)).
\end{equation}

\begin{theorem}
\label{thm:wmem}
For $|A|>1$ an HMA that characterizes
$\CP_{\wmem}$ exists, i.e. there is an
HMA $\BA^{\wmem}$ such that for all 
$t,t'\in\NT_{\SigmaHMA}$,
$\CP_{\wmem}\vdash t=t'\iff \BA^{\wmem}\models t=t'$.
\end{theorem}

\begin{proof}
Construct the $\SigmaStateful$-algebra $\NS^{\wmem}$
with
$\NT_{\SigmaHMA}/_{=_{\wmem}}$ as its set of
conditional expressions and 
$\WM$ (Definition~\ref{def:WM})
as its set of states. 
We first argue that $\WM$
is suitable as state set.
Define for $f\in \WM$,
$a\reply f=f(a)$ and 
for $\sigma\in A^\con$, 
\[(a\apply f)(\sigma)=
f(a\lcon \sigma).
\]
This is well-defined: if $f\in\WM$ then it easily follows
that for all
$a\in A$, $a\apply f\in \WM$. We note that for all
$a\in A$ and $f\in WM$, $a\apply(a\apply f)=a\apply f$, and
also
\begin{equation}
\label{eq:vbwm2}
f(a)=f(ab)~\Longrightarrow~ a\apply(b\apply(a\apply f))=
b\apply(a\apply f).
\end{equation}
The latter conditional equation follows immediately from
Definition~\ref{def:WM}.

Similar as in the proof of Theorem~\ref{thm:1},
the state constant $c$ is given an arbitrary interpretation, and
the axioms~\eqref{2S1}--\eqref{2S3}
define the function $s\lef f\rig s'$ in $\NS^{\wmem}$.
The axioms~\eqref{SPA1}--\eqref{SPA6}
fully determine the functions $\reply$ and $\apply$, and this
is well-defined:
if $t=_{\wmem}t'$ then for all $f\in\WM$, 
$t\reply f=t'\reply f$ and $t\apply f=t'\apply f$
follow by inspection of the 
$\CP_{\wmem}$ axioms. We show soundness of 
the axiom \eqref{CPwm1}.
Assume $a\ne b$ and
$f(a)=f(ab)$,
then $f(aba)=f(a)$ and by equation~\eqref{eq:vbwm2} (case $f(a)=\tr$),
\begin{align*}
&(((t_1\lef a\rig t_2)\lef b\rig t_3)\lef a\rig t)\reply f\\
&=[(t_1\lef a\rig t_2)\reply(b\apply(a\apply f))\lef b\reply (a\apply f)
\rig t_3\reply(b\apply(a\apply f))]\lef a\reply f\rig t\reply(a\apply f)\\
&=[t_1\reply(a\apply(b\apply(a\apply f)))\lef b\reply(a\apply f)\rig 
t_3\reply(b\apply(a\apply f))]\lef a\reply f\rig t\reply(a\apply f)\\
&=[t_1\reply(b\apply(a\apply f))\lef b\reply(a\apply f)\rig t_3\reply
(b\apply(a\apply f))]\lef a\reply f\rig t\reply(a\apply f)\\
&=((t_1\lef b\rig t_3)\lef a\rig t)\reply f,
\end{align*}
and in a similar way
$(((t_1\lef a\rig t_2)\lef b\rig t_3)\lef a\rig t)\apply f
=(t_1\lef a\rig t_3)\apply f$ follows.
The cases $a=b$ and $f(a)\ne f(ab)$ are trivial.
Soundness of 
the axiom \eqref{CPwm2} follows in a similar way.
The axiom~\eqref{SPA7} holds by construction of \RP.
In order to prove that $\NS^{\wmem}$ is a SPA 
it remains to be shown that 
axiom~\eqref{SPA8} holds, i.e.,
for all $t,t'\in\NT_{\SigmaHMA}$,
\[\forall f(t\reply f=t'\reply f\wedge
t\apply f=t'\apply f)\rightarrow t=_{\wmem}t'.\]
This follows by contraposition. We may assume that
both $t$ and $t'$ are \wmem-basic forms, and we 
apply induction on the
complexity of $t$. 
Let $a,b\in A$ with $a\ne b$.
\begin{enumerate}
\item
If $t\equiv \tr$, then if $t'\equiv\fa$ 
it follows that $t\reply f
\ne t'\reply f$ for any $f$,  and if 
$t'\equiv t_1\lef a\rig t_2$ then consider 
some $f$ with $f(a)=f(b)=\tr$ and 
$f(ab)=f(a\sigma b)=\fa$ for all appropriate
$\sigma$. We find $(t\apply f)(b)=f(b)=\tr$ 
and
$(t'\apply f)(b)=(t_1\apply f)(a\sigma b)=\fa$ 
(where $\sigma$ possibly equals $\epsilon$), so $t'\apply f\ne t\apply f$.

\item
If $t\equiv\fa$ a similar argument applies.

\item
If $t\equiv t_1\lef a\rig t_2$, then the case 
$t'\in\{\tr,\fa\}$
can be dealt with as above. 

If $t'\equiv t_3\lef a\rig t_4$ then assume
$t\ne_\wmem t'$ because $t_1\ne_\wmem t_3$.
Then $a\not\in pos(t_1)\cup pos(t_3)$, and by induction there is $f$ with
$t_1\apply f\neq t_3\apply f$ or
$t_1\reply f\neq t_3\reply f$.
Take $g$ such that $a\apply g=f$ and $a\reply g=\tr$,
then $g$ distinguishes
$t_1\lef a\rig t_2$ and $t_3\lef a\rig t_4$. Note that
the restriction obtained by $a\apply g=f$ and $a\reply g=\tr$
that is imposed by Definition~\ref{def:WM}, i.e.,
for all $b\in A\setminus\{a\}$,
\[g(a b)=g(a)\Longrightarrow
\begin{cases}
g(aba)=g(a),
\text{ and}\\
g(aba\lcon\sigma)=g(ab\lcon\sigma)&\text{for all
$\sigma\in A^\con$}.
\end{cases}
\]
is not relevant because of $a\not\in pos(t_1)\cup pos(t_3)$,
and hence values of $g(aba\lcon\sigma)$ play not a role
in the above-mentioned distinction.

If $t_1=_{\wmem}t_3$, then a similar argument applies
for $t_2\neq_{\wmem}t_4$.

If $t'\equiv t_3\lef b\rig t_4$ then 
$(t_1\lef a\rig t_2)\apply f\ne (t_3\lef b\rig t_4)\apply f$ 
for $f$ defined by $f(a)=f(a\sigma)=\tr$ and 
$f(b)=f(b\sigma')=\fa$ for all appropriate
$\sigma,\sigma'$ because
$((t_1\lef a\rig t_2)\apply f)(b)=(t_1\apply (a\apply f))(b)=
f(a\rho b)=\tr$ (where
$\rho$ possibly equals $\epsilon$), and
$((t_3\lef b\rig t_4)\apply f)(b)=(t_4\apply (b\apply f))(b)$
and this equals either
$f(b\rho' b)=\fa$ for
some $\rho'\ne\epsilon$, or $f(b)=\fa$.
\end{enumerate}
So $\NS^{\wmem}$ is a SPA. 
Define the HMA $\BA^{\wmem}$ as the $\SigmaHMA$-reduct 
of $\NS^{\wmem}$. The above argument on the soundness of
the axiom schemes~\eqref{CPwm1} and \eqref{CPwm2}
proves $\Longrightarrow$ as stated in the theorem,
and the
validity of axiom~\eqref{SPA8} proves $\Longleftarrow$.
We finally note that 
$\BA^{\wmem}\cong I(\SigmaHMA,\CP_{\wmem})$.
\end{proof}

In the proof above we defined
the SPA $\NS^{\wmem}$ and we found that if
$|A|>1$, then
for all $t,t'\in\NT_{\SigmaHMA}$,
\begin{equation}
\label{eq:cwm}
\CP_{\wmem}\vdash t=t' \iff \NS^{\wmem}\models t=t'.
\end{equation}
If $A=\{a\}$ then $A^{\con}=A$ and 
$\NS^{\wmem}$ as defined above has only two states,
say $f$ and $g$ with
$f(a)=\tr$ and $g(a)=\fa$. It easily follows that 
\[\BA^{\wmem}\models \tr\lef a\rig\tr=\tr,\]
so 
$\BA^{\wmem}\not\cong I(\SigmaHMA,\CP_{\wmem})$ if $A=\{a\}$.
The following corollary is related to Theorem~\ref{thm:wmem} 
and characterizes weakly memorizing congruence in terms
of a quasivariety of SPAs that satisfy two extra conditions.

\begin{corollary}
\label{cor:wmem}
Let $|A|>1$. 
Let $\NC_{\wmem}$ be the class of SPAs that satisfy
for all $a,b\in A$ and $s\in S$,
\begin{align*}
&a\reply (a\apply s)= a\reply s~\wedge~
a\apply (a\apply s)= a\apply s,\\
&b\reply(a\apply s)=a\reply s\rightarrow
(a\reply(b\apply (a\apply s))=a\apply s ~\wedge~
a\apply(b\apply(a\apply s))=b\apply(a\apply s)).
\end{align*}
Then for all $t,t'\in\NT_{\SigmaHMA}$,
\[\NC_{\wmem}\models t=t'\iff \CP_{\wmem}\vdash t=t'.\]
\end{corollary}

\begin{proof}
By its definition, $\NS^{\wmem}\in \NC_{\wmem}$, 
which by \eqref{eq:cwm} implies $\Longrightarrow$.
For the converse, it is sufficient to show that
the axioms~\eqref{CPwm1} and \eqref{CPwm2} 
hold in each SPA 
that is in $\NC_{\wmem}$ because $\NC_{\wmem}\subseteq
\NC_\con$.
Let such $\NS$ be given.
Consider \eqref{CPwm1}: 
if for some
interpretation of $s$ in \NS, $a\reply s=\fa$ 
there is nothing to prove, and if 
$a\reply s=b\reply(a\apply s)=\tr$ and thus
$a\reply (b\apply(a\apply s))=\tr$, then
\begin{align*}
(((t_1\lef a\rig t_2)\lef b\rig t_3)\lef a\rig t)\reply s&=
t_1\reply
(a\apply(b\apply (a\apply s)))\\
&=t_1\reply (b\apply(a\apply s))\\
&=((t_1\lef b\rig t_3)\lef a\rig t)\reply s,
\end{align*}
and
\begin{align*}
(((t_1\lef a\rig t_2)\lef b\rig t_3)\lef a\rig t)\apply s&=
t_1\apply
(a\apply (b\apply(a\apply s)))\\
&=t_1\apply (b\apply(a\apply s))\\
&=((t_1\lef b\rig t_3)\lef a\rig t)\apply s.
\end{align*}
The soundness of axiom 
\eqref{CPwm2} can be proved
in a similar way.
\end{proof}

\section{Memorizing  congruence}
\label{sec:mem}
In this section we consider \emph{memorizing congruence}.
We define $\CP_{\mem}$ as the extension of 
\CP\ with the axiom  
\begin{align*}
\label{CPmem}\tag{CPmem}
x\lef y\rig(z\lef u\rig(v\lef y\rig w))&= x\lef y\rig(z\lef u\rig w).
\end{align*}

Axiom \eqref{CPmem}
defines how the central condition $y$ may recur in a
propositional statement, and thus defines a general
form of contraction.
The symmetric variants of \eqref{CPmem}, i.e.,
\begin{align}
\label{eq:mini}
x\lef y\rig ((z\lef y\rig u)\lef v\rig w)&=x\lef y\rig (u\lef v\rig w),\\
\label{eq:mini1}
(x\lef y\rig (z\lef u\rig v)) \lef u\rig w&=(x\lef y\rig z)\lef u\rig w,\\
\label{eq:mini2}
((x\lef y\rig z)\lef u\rig v) \lef y\rig w&=(x\lef u\rig v)\lef y\rig w,
\end{align} 
all follow easily with 
$y\lef x\rig z=z\lef(\fa\lef x\rig\tr)\rig y$ (which is derivable in \CP), 
e.g., a proof of~\eqref{eq:mini} is as follows:
\begin{align*}
x\lef y\rig ((z\lef y\rig u)\lef v\rig w)
&=x \lef y\rig (w\lef(\fa\lef v\rig\tr)\rig(z\lef y\rig u))\\
&=x \lef y\rig (w\lef(\fa\lef v\rig\tr)\rig u)\\
&=x\lef y\rig (u\lef v\rig w).
\end{align*}

Let \emph{memorizing congruence}, notation $=_{\mem}$,
be the congruence on $\NT_{\SigmaHMA}$ generated 
by the axioms of $\CP_{\mem}$.  

\begin{definition}
\label{def:membf}
A term $t\in\NT_{\SigmaStateless}$ is a
\textbf{mem-basic form over $A'\subset A$}
if for $a\in A'$,
\[t::=\tr\mid\fa\mid t_1\lef a \rig t_2 \]
and $t_i$ ($i=1,2$) is a \mem-basic form 
over $A'\setminus\{a\}$.
\end{definition}

E.g., for $A=\{a\}$ the set of all \mem-basic forms is 
$\{B,\;B\lef a\rig B'\mid 
B,B'\in\{\tr,\fa\}\}$, and for $A=\{a,b\}$ it is
\begin{align*}
\{B,\;t_1\lef a\rig t_2,t_3\lef b\rig t_4\mid~&
B\in\{\tr,\fa\},\\
&t_1,t_2 \text{ \mem-basic forms over $\{b\}$, }\\
&
t_3,t_4 \text{ \mem-basic forms over $\{a\}$}\}.
\end{align*} 
For $|A|=n$, the number of \mem-basic forms is 
$a_n = n(a_{n-1})^2 + 2$ with $a_0=2$, so the first few
values are $6, 74, 16430$.

\begin{lemma}
\label{lem:mem}
For each $t\in\NT_{\SigmaHMA}$ there exists a \mem-basic form
$t'$ with $\CP_{\mem}\vdash t=t'$.
\end{lemma}
\begin{proof}
See~\cite{BP10}; this proof is repeated in
Appendix~\ref{app:Proofs}.
\end{proof}

\begin{definition}
\label{def:Acore}
Let $A^{core}\subset A^+$ be the set of strings 
in which each element of $A$ occurs at most 
once.\footnote{If $|A|=n$ then $|A^{core}|=b_n$ 
  with $b_1=1$ and
  $b_n=n(b_{n-1}+1)$. (The first few $b_n$-values are
  $1,4,15,64,325,...$).}
\end{definition}
  
 We first argue that $M=\{\tr,\fa\}^{A^{core}}$
is suitable as state set of a SPA that characterizes
$\CP_{\mem}$.
Define for $f\in M$ the following:
$a\reply f=f(a)$ and 
for $\sigma\in A^{core}$, 
\[(a\apply f)(\sigma)=
\begin{cases}
f(a)&\text{if $\sigma=a$ or $\sigma=\rho a$},\\
f(a(\sigma-a))&\text{otherwise, where 
$(\sigma-a)$ is as $\sigma$ but with $a$ left out}.
\end{cases}
\]
For example, $(a\apply )f(a)=(a\apply f)(ba)=f(a)$ and 
$(a\apply f)(b)=(a\apply f)(ab)=f(ab)$.
Observe that 
\[(t'\lef t\rig t')\apply f=t'\apply (t\apply f)\]
because
if $t\reply f=\tr$ then 
$(t'\lef t\rig t')\apply f=t'\apply (t\apply f)$ 
and this also holds if $t\reply f=\fa$; now apply
axiom~\eqref{SPA7}.

\begin{lemma}
\label{lem:mem2}
For all $f\in M$ and
$t,t'\in\NT_{\SigmaHMA}$,
\begin{equation}
\label{eq:prof}
t\reply(t'\apply(t \apply f))=t\reply f~\wedge~ 
t\apply(t'\apply(t \apply f))=t'\apply(t\apply f).
\end{equation}
\end{lemma}
\begin{proof}
See Appendix~\ref{app:Proofs}.
\end{proof}

\begin{theorem}
\label{thm:3}
For $|A|>1$ an HMA that characterizes $\CP_{\mem}$ 
exists, i.e. there is an
HMA $\BA^{\mem}$ such that for all 
$t,t'\in\NT_{\SigmaHMA}$,
$\CP_{\mem}\vdash t=t'\iff \BA^{\mem}\models t=t'$.
\end{theorem}

\begin{proof}
Construct the $\SigmaStateful$-algebra $\NS^{\mem}$
with
$\NT_{\SigmaHMA}/_{=_{\mem}}$ as the set of
conditional expressions and the function space
$M$ as defined 
above as the set of states. Furthermore, adopt
the definitions of $a\reply f$ and $a\apply f$
given above.

Similar as in the proof of Theorem~\ref{thm:1},
the state constant $c$ is given an arbitrary interpretation, and
the axioms~\eqref{2S1}--\eqref{2S3}
define the function $s\lef f\rig s'$ in $\NS^{\mem}$.
The axioms~\eqref{SPA1}--\eqref{SPA6}
fully determine
the functions $\reply$ and $\apply$, and this is well-defined:
if $t=_{\mem}t'$ then for all $f$, $t\reply f=t'\reply f$ and
$t\apply f=t'\apply f$.
We show soundness of the axiom \eqref{CPmem}:
consider an arbitrary closed
instance $t_1\lef t_2\rig(t_3\lef t_4\rig(t_5\lef t_2\rig t_6))
= t_1\lef t_2\rig(t_3\lef t_4\rig t_6)$. 
A sufficient property to conclude for all states $f$ that
\begin{align*}
(t_1\lef t_2\rig(t_3\lef t_4\rig(t_5\lef t_2\rig t_6)))
\reply f&= (t_1\lef t_2\rig(t_3\lef t_4\rig t_6))\reply f,
\\
(t_1\lef t_2\rig(t_3\lef t_4\rig(t_5\lef t_2\rig t_6)))
\apply f&= (t_1\lef t_2\rig(t_3\lef t_4\rig t_6))\apply f
\end{align*}
is the validity of equation \eqref{eq:prof}
(read $t_2$ for $t$ and $t_4$ for $u$), which was
proved in Lemma~\ref{lem:mem2}.
The axiom~\eqref{SPA7} holds by construction of \RP.
In order to prove that $\NS^{\mem}$ is a SPA 
it remains to be shown that 
axiom~\eqref{SPA8} holds, i.e.,
for all $t,t'\in\NT_{\SigmaHMA}$,
\[\forall f(t\reply f=t'\reply f\wedge
t\apply f=t'\apply f)\rightarrow t=_{\mem}t'.\]
This follows by contraposition. We may assume that
both $t$ and $t'$ are \mem-basic forms, and we 
apply induction on the
complexity of $t$. 
Let $a,b\in A$ with $a\ne b$.
\begin{enumerate}
\item
If $t\equiv \tr$, then if $t'\equiv\fa$ 
it follows that $t\reply f
\ne t'\reply f$ for any $f$,  and if 
$t'\equiv t_1\lef a\rig t_2$ then consider 
some $f$ with $f(a)=f(b)=\tr$ and 
$f(ab)=f(a\sigma b)=\fa$ for all appropriate
$\sigma$. We find $(t\apply f)(b)=f(b)=\tr$ 
and
$(t'\apply f)(b)=(t_1\apply f)(a\sigma b)=\fa$ 
(where $\sigma$ possibly equals $\epsilon$), 
so $t'\apply f\ne t\apply f$.

\item
If $t\equiv\fa$ a similar argument applies.

\item
If $t\equiv t_1\lef a\rig t_2$, then the case 
$t'\in\{\tr,\fa\}$
can be dealt with as above. 

If $t'\equiv t_3\lef a\rig t_4$ then assume
$t\ne_\mem t'$ because $t_1\ne_\mem t_3$.
Then $a$ does not occur in any of the $t_i$, and by induction there is $f$ with
$t_1\apply f\neq t_3\apply f$ or
$t_1\reply f\neq t_3\reply f$.
Take $g$ such that $g\restriction A\setminus\{a\}=
f\restriction \setminus\{a\}$ and $a\apply g=f$ and $a\reply g=\tr$,
then $g$ distinguishes
$t_1\lef a\rig t_2$ and $t_3\lef a\rig t_4$. 
If $t_1=_{\mem}t_3$, then a similar argument applies
for $t_2\neq_{\mem}t_4$.

If $t'\equiv t_3\lef b\rig t_4$ then
$(t_1\lef a\rig t_2)\apply f\ne (t_3\lef b\rig t_4)\apply f$ 
for $f$ defined by $f(a)=f(a\sigma)=\tr$ and 
$f(b)=f(b\sigma')=\fa$ for all appropriate
$\sigma,\sigma'$ because
$((t_1\lef a\rig t_2)\apply f)(b)=(t_1\apply (a\apply f))(b)=
f(a\rho b)=\tr$ (where
$\rho$ possibly equals $\epsilon$), and
$((t_3\lef b\rig t_4)\apply f)(b)=(t_4\apply (b\apply f))(b)$
and this equals either
$f(b\rho' )=\fa$ for
some $\rho'\ne\epsilon$, or $f(b)=\fa$.
\end{enumerate} 
So $\NS^{\mem}$ is a SPA. 
Define the HMA $\BA^{\mem}$ as the $\SigmaHMA$-reduct 
of $\NS^{\mem}$. The above argument on the soundness of
the axiom \eqref{CPmem}
proves $\Longrightarrow$ as stated in the theorem,
and the
validity of axiom~\eqref{SPA8} proves $\Longleftarrow$.
Observe that
$\BA^{\mem}\cong I(\SigmaHMA,\CP_{\mem})$.
\end{proof}

\begin{remark}
If $A=\{a\}$ then $\NS^{\mem}$ as defined
above has only two states,
say $f$ and $g$ with
$f(a)=\tr$ and $g(a)=\fa$. It then easily follows that 
$\BA^{\mem}\models \tr\lef a\rig\tr=\tr$ so in that case
$\BA^{\mem}\not\cong I(\SigmaHMA,\CP_{\mem})$. 

Furthermore, if
$A=\{a,b\}$ it easily follows that $\NS^{\mem}\not\models
a\leftand b=
b\leftand a$: take $f$ such that $f(a)=f(ab)=\tr$ and
$f(b)=\fa$.
\end{remark}

If $|A|>1$, then 
it follows from the proof of Theorem~\ref{thm:3} that
for all
$t,t'\in\NT_{\SigmaHMA}$,
\begin{equation}
\label{eq:cmem}
\CP_{\mem}\vdash t=t' \iff \NS^{\mem}\models t=t'.
\end{equation}

The following corollary is related to Theorem~\ref{thm:3} 
and characterizes memorizing congruence in terms
of a quasivariety of SPAs that satisfy an extra condition.
\begin{corollary}
\label{cor:mem}
Let $|A|>1$. Let
$\NC_{\mem}$ be the class of SPAs that satisfy
for all $a\in A$ and $s\in S$,
\begin{align}
\label{eq:mem}
a\reply(x\apply( a\apply s))=a\reply s~\wedge~&
a\apply (x\apply (a\apply s))=x\apply (a\apply s).
\end{align}
(Note that with $x=\tr$ this yields the axiom scheme from 
Corollary~\ref{cor:cr} that characterizes contractive congruence.)
Then for all $t,t'\in\NT_{\SigmaHMA}$,
\[\NC_{\mem}\models t=t'\iff \CP_{\mem}\vdash t=t'.\]
\end{corollary}

\begin{proof}
By its definition we find that $\NS^{\mem}\in \NC_{\mem}$, 
which by \eqref{eq:cmem} implies $\Longrightarrow$.
For the converse, it is sufficient to show that
the axiom
\eqref{CPmem} holds in each SPA in $\NC_{\mem}$. 
Let such $\NS$ be given.
Consider an arbitrary closed
instance $t_1\lef t_2\rig(t_3\lef t_4\rig(t_5\lef t_2\rig t_6))
= t_1\lef t_2\rig(t_3\lef t_4\rig t_6)$. 
A sufficient property to conclude that
\begin{align*}
(t_1\lef t_2\rig(t_3\lef t_4\rig(t_5\lef t_2\rig t_6)))
\reply s&= (t_1\lef t_2\rig(t_3\lef t_4\rig t_6))\reply s,
\\
(t_1\lef t_2\rig(t_3\lef t_4\rig(t_5\lef t_2\rig t_6)))
\apply s&= (t_1\lef t_2\rig(t_3\lef t_4\rig t_6))\apply s
\end{align*}
is the following (read $t_2$ for $t$ and $t_4$ for $t'$): 
\begin{equation}
\label{eq:pro}
t\reply(t'\apply(t \apply s))=t\reply s
~\wedge~
t\apply (t'\apply (t\apply s))=t'\apply (t\apply s).
\end{equation}
We prove this property by structural induction on $t$. 
If $t\equiv\tr$ or
$t\equiv\fa$ or $t\equiv a\in A$ then \eqref{eq:pro} follows immediately.
If $t \equiv t_1\lef t_2\rig t_3$ we make
a case distinction:
\\[2mm]
$(i)$ Assume for some interpretation of $s$ in \NS,
$t_2\reply s=\tr$.
We derive
$t_2\reply(t'\apply(t_1\apply (t_2\apply s)))=
t_2\reply((t'\lef t_1\rig t')\apply (t_2\apply s))$ and by 
the induction hypothesis (IH) we find
$t_2\reply((t'\lef t_1\rig t')\apply (t_2\apply s))=
t_2\reply s=\tr$. We further derive
\begin{align*}
t\reply(t'\apply(t \apply s))&=t\reply(t'\apply(t_1\apply (t_2\apply s)))
\\
&=(t_1\lef t_2\rig t_3)\reply(t'\apply(t_1\apply (t_2\apply s)))\\
&=t_1\reply(t_2\apply(t'\apply(t_1\apply (t_2\apply s))))
\\
&=t_1\reply((t_2\lef t'\rig t_2)\apply(t_1\apply (t_2\apply s)))
\\
&=t_1\reply(t_2\apply s)&&\text{(by IH)}\\
&=t\reply s,
\end{align*}
and
\begin{align*}
t\apply(t'\apply(t \apply s))&=t\apply(t'\apply(t_1\apply (t_2\apply s)))
\\
&=(t_1\lef t_2\rig t_3)\apply(t'\apply(t_1\apply (t_2\apply s)))\\
&=t_1\apply(t_2\apply(t'\apply(t_1\apply (t_2\apply s))))
\\
&=t_1\apply((t_2\lef t'\rig t_2)\apply(t_1\apply (t_2\apply s)))
\\
&=(t_2\lef t'\rig t_2)\apply(t_1\apply (t_2\apply s))
&&\text{(by IH)}\\
&=t_2\apply( t'\apply(t_1\apply (t_2\apply s)))\\
&=t_2\apply( (t'\lef t_1\rig t')\apply (t_2\apply s))\\
&=(t'\lef t_1\rig t')\apply(t_2\apply s)&&\text{(by IH)}\\
&=t'\apply(t_1\apply(t_2\apply s))\\
&=t'\apply(t\apply s).
\end{align*}
$(ii)$ Assume for some interpretation of $s$ in \NS,
$t_2\reply s=\fa$. Similar.
\end{proof}

\section{Static congruence (Propositional logic)}
\label{sec:stat}
In this section we consider \emph{static congruence}
defined by the axioms of $\CP$ and the axioms
\begin{align*}
\label{CPstat}\tag{CPstat} \qquad
(x\lef y\rig z)\lef u\rig v&=
(x\lef u\rig v)\lef y\rig(z\lef u\rig v),\\
\label{CPcontr}\tag{CPcontr} \qquad
(x\lef y\rig z)\lef y\rig u&= x \lef y\rig u.
\end{align*}
We write $\CP_{\stat}$ for this set of axioms.
Note that the symmetric variants of 
the axioms~\eqref{CPstat} and \eqref{CPcontr}, say
\begin{align*}
\label{CPstat'}\tag{CPstat$'$} \qquad
x\lef y\rig (z\lef u\rig v)&=
(x\lef y\rig z)\lef u\rig(x\lef y\rig v),\\
\label{CPcontr'}\tag{CPcontr$'$} \qquad
x\lef y\rig(z\lef y\rig u) &= x \lef y\rig u,
\end{align*}
easily follow with identity  
$y\lef x\rig z=(z\lef\fa\rig y)\lef x\rig (z\lef\tr\rig y)=
z\lef (\fa\lef x\rig\tr)\rig y$
(thus an
identity derivable in \CP). 
Moreover, in $\CP_{\stat}$ it follows that
\begin{align*}
x
&=(x\lef y\rig z)\lef \fa\rig x\\
&=(x\lef\fa\rig x)\lef y\rig (z\lef \fa\rig x)\\
&=x\lef y\rig x
&&\text{(cf.\ equation~\eqref{X}).}
\end{align*}

We define \emph{static congruence} $=_{\stat}$ on $\NT_{\SigmaHMA}$ 
as the congruence generated by $\CP_{\stat}$.
Let $t,t'\in\NT_{\SigmaHMA}$.
Then under static congruence, $t$ and $t'$
can be rewritten into the following
special type of basic form: assume the atoms 
occurring in $t$ and $t'$ are $a_1,...,a_n$, and
consider the full binary tree with at level 
$i$ only occurrences of atom $a_i$ (there are $2^{i-1}$ such
occurrences), 
and at level $n+1$ only leaves 
that are either $\tr$ or $\fa$
(there are $2^n$ such leaves). 
For example, for $n=2$ we find
the $2^4$ different terms
\[(\tr/\fa\lef a_2\rig \tr/\fa)\lef a_1\rig (\tr/\fa\lef a_2\rig \tr/\fa).\]
Then the axioms in
$\CP_{\stat}$
are sufficient to rewrite both $t$ and $t'$ into
exactly one such special basic form. 

\begin{theorem}
\label{thm:2}
There exists an HMA that characterizes propositional logic, i.e. there is an
HMA $\BA^{\stat}$ such that for all 
$t,t'\in\NT_{\SigmaHMA}$,
$\CP_{\stat}\vdash t=t'\iff \BA^{\stat}\models t=t'$.
\end{theorem}

\begin{proof}
Construct the $\SigmaStateful$-algebra $\NS^{\stat}$
with
$\NT_{\SigmaHMA}/_{=_{\stat}}$ as the set of
conditional expressions and the function space
$\{\tr,\fa\}^{A}$ as the set of states. For 
each state $f$
and atom $a\in A$ define 
$a\reply f=f(a)$ and $a\apply f=f$.
Similar as in the proof of Theorem~\ref{thm:1},
the state constant $c$ is given an arbitrary interpretation, and
the axioms~\eqref{2S1}--\eqref{2S3}
define the function $s\lef f\rig s'$ in $\NS^{\stat}$.
The axioms~\eqref{SPA1}--\eqref{SPA6}
fully determine
the functions $\reply$ and $\apply$, and this is well-defined:
if $t=_{\stat}t'$ then for all $f$, $t\reply f=t'\reply f$
and
$t\apply f=t'\apply f$ follow by inspection of the 
$\CP_{\stat}$
axioms.
The axiom~\eqref{SPA7} holds by construction of \RP.
In order to prove that $\NS^{\stat}$ is a SPA 
it remains to be shown that 
axiom~\eqref{SPA8} holds, i.e.,
for all $t,t'\in\NT_{\SigmaHMA}$,
\[\forall f(t\reply f=t'\reply f\wedge
t\apply f=t'\apply f)\rightarrow t=_{\stat}t'.\]
This follows by contraposition. We may assume that
both $t$ and $t'$
are in the basic form described above: if $t$ and $t'$ are
different in some leaf then the reply function $f$
leading to this leaf satisfies $t\reply f\neq t'\reply f$.

So $\NS^{\stat}$ is a SPA. 
Define the HMA $\BA^{\stat}$ as the $\SigmaHMA$-reduct 
of $\NS^{\stat}$. The above argument on the soundness of
the axioms~\eqref{CPstat} and \eqref{CPcontr}
proves $\Longrightarrow$ as stated in the theorem,
and the
validity of axiom~\eqref{SPA8} proves $\Longleftarrow$.
Moreover, $\BA^{\stat}\cong I(\SigmaHMA,\CP_{\stat})$.
\end{proof}

From the proof above it follows that for all
$t,t'\in\NT_{\SigmaHMA}$,
\begin{equation}
\label{eq:cc}
\CP_{\stat}\vdash t=t' \iff \NS^{\stat}\models t=t'.
\end{equation}

\begin{corollary}
\label{cor:stat}
Let $\NC_{\stat}$ be the class of SPAs that satisfy
for all $a\in A$ and $s\in S$,
\[a\apply s= s.\]
Then for all $t,t'\in\NT_{\SigmaHMA}$,
\[\NC_{\stat}\models t=t'\iff \CP_{\stat}\vdash t=t'.\]
\end{corollary}

\begin{proof}
By its definition, $\NS^{\stat}\in \NC_{\stat}$, 
which by \eqref{eq:cc} implies $\Longrightarrow$.
For the converse, it is sufficient to show that
the axioms
\eqref{CPstat} and \eqref{CPcontr} hold in each
SPA in $\NC_{\stat}$. 
We first prove by structural induction on 
$t\in\NT_{\SigmaHMA}$ the $\NC_{\stat}$-identity
\[t\apply s=s.\]
If $t\in\{\tr,\fa,a\mid a\in A\}$ this is clear, and if 
$t\equiv t_1\lef t_2\rig t_3$ then
\begin{align*}
t\apply s&=(t_1\lef t_2\rig t_3)\apply s\\
&=(t_1\apply(t_2\apply s))\lef t_2\reply s\rig
(t_3\apply(t_2\apply s))\\
&=(t_1\apply s)\lef t_2\reply s\rig
(t_3\apply s)&&\text{(by IH)}\\
&=s\lef t_2\reply s\rig s&&\text{(by IH)}\\
&=s.
\end{align*}
With the identity $t\apply s=s$ the soundness of the axioms
\eqref{CPstat} and \eqref{CPcontr} follows easily:
let $\NS\in\NC_{\stat}$ be given.
Consider a closed instance of~\eqref{CPstat}:
\[(t_1\lef t_2\rig t_3)\lef t_4\rig t_5=(t_1\lef t_4\rig t_5)
\lef t_2\rig (t_3\lef t_4\rig t_5).\]
Then for all states $s$,
both the left-hand side and the right-hand side 
transform $s$ under $\apply$ to $s$, so
\[((t_1\lef t_2\rig t_3)\lef t_4\rig t_5)\reply s=
(t_1\reply s\lef
t_2\reply s\rig t_3\reply s)
\lef t_4\reply s\rig t_5\reply s\]
and
\[((t_1\lef t_4\rig t_5)\lef t_2\rig (t_3\lef t_4\rig t_5))\reply s
=(t_1\reply s\lef t_4\reply s\rig t_5\reply s)
\lef t_2\reply s\rig (t_3\reply s\lef t_4\reply s\rig t_5\reply s).
\]
By case distinction on the reply values of $t_4$ and $t_2$ in \NS,
it easily follows that both these instances
yield equal values. 
The soundness of axiom \eqref{CPcontr} can be proved
in the same way.
\end{proof}

\section{Conclusions and related work}
\label{sec:Con}
A main result in our defining paper on proposition
algebra~\cite{BP10} concerns its semantics: in
that paper we define
valuation algebras (VAs) as two-sorted algebras with the
Boolean constants and valuations as their sorts. Using 
these, valuation varieties (varieties of VAs) are defined by
equational specifications. For example, the free variety
\fr\ contains all VAs, and the variety \rp\ of 
repetition-proof VAs is the subvariety of VAs that
satisfy the axiom (in the notation of this paper)
\[a\reply (a\apply s)=a\reply s\]
(cf.\ Corollary~\ref{cor:rp}). A valuation variety 
defines a valuation equivalence by identifying all
propositional statements that yield the same evaluation
result in all VAs in that variety. 
For example, $\tr$ and $\tr\lef a\rig\tr$ are
valuation equivalent in all valuation varieties we consider.
For
\[K\in\{\fr,\rp,\con,\wmem,\mem,\stat\},\]
such a valuation equivalence is denoted by $\equiv_K$,
and |
overloading notation here | the valuation
congruence $=_K$ is defined as
the largest congruence contained in $\equiv_K$.
We prove that for all $t,t'\in\NT_{\SigmaHMA}$,
\[\CP_K\vdash t=t'\iff t=_K t',\]
where $\CP_\fr$ denotes the axiom set \CP.\footnote{In
  this paper we use the notation $t=_Kt'$ as a 
  shorthand for $\CP_K\vdash t=t'$, but according
  to the above-mentioned result this overloading
  is not a problem. }
  
In this paper we provide an alternative semantics
for proposition algebra in the form of HMAs, which
has the advantage that we can define a valuation congruence
without first defining some 
valuation equivalence it is contained in. 
Our HMA-semantics provides by construction a valuation
congruence and the relation between evaluation of propositions
and transformation of valuations appears to be more 
elegant. A typical difference between semantics based on
VAs and our semantics based on HMAs is that the apply
operator $\apply$ in the latter is defined
on a more general level. We see this difference
if we compare the definition of the variety of VAs that
defines static congruence with the quasivariety of SPAs
that characterizes $\CP_\stat$:
in the former, the crucial axiom on valuations
reads as follows: for all atoms $a,b\in A$ and
valuations $s$,
\[a\reply(b\apply s)=a\reply s,\]
while according to Corollary~\ref{cor:stat}, HMA-semantics 
requires in the case of static congruence that
for all atoms $a\in A$ and
valuations $s$,
\[a\apply s=s.\footnote{This is the case
because if $\CP_\stat\vdash t=t'$,
then for each
SPA that characterizes static valuation congruence
it should be the case that
both $t\reply s=t'\reply s$ and 
$t\apply s =t'\apply s$ hold 
(and therefore $t=t'$ by 
axiom~\eqref{SPA8}); now 
observe that $t\apply s =t'\apply s$
does not follow from \CPS\ extended
with the weaker requirement
$a\reply(b\apply s)=a\reply s$ 
(for all atoms $a,b\in A$ and
valuations $s$).
}\]
A question related to the difference
between VA-based and HMA-based semantics
is to either prove or refute that the class
$\NC_{\mem}$ (see Corollary~\ref{cor:mem})
is definable by
weakening requirement~\eqref{eq:mem} on its SPAs to 
this one:
for all $a,b\in A$ and $s\in S$,
\begin{align}
\label{eq:memcr}
a\reply (a\apply s)= a\reply s~\wedge~&
a\apply (a\apply s)= a\apply s,\\
\label{eq:mem2}
a\reply(b\apply( a\apply s))=a\reply s~\wedge~&
a\apply (b\apply (a\apply s))=b\apply (a\apply s),
\end{align}
because equations~\eqref{eq:memcr} and \eqref{eq:mem2} 
exactly capture the
variety of VAs that characterizes memorizing
valuation congruence (cf.\ \cite{BP10}).
Last but not least, a semantics for proposition algebra
based on HMAs refutes
axiomatizations such as the one defined by 
$\CP + \langle\, \tr\lef x\rig \tr=\tr\,\rangle$,
which indeed is a 
peculiar axiomatization if one analyzes it in 
terms of TRSs (see Theorem~\ref{thm:oh} in 
Section~\ref{sec:2}).

Further results from~\cite{BP10} 
concern binary connectives: we prove that the
conditional connective cannot be expressed modulo 
$=_\con$ (or any finer congruence) if only binary 
connectives are allowed, but that it can be expressed
modulo $=_\mem$ (and $=_\stat$); for $=_\wmem$ we leave this question
open.
In the papers~\cite{BP10,BP10a}
we use the notation $\leftand$ (taken from~\cite{BBR95})
for left-sequential conjunction, defined by
\[x\leftand y=y\lef x\rig\fa,\]
and elaborate on the connection between sequential
binary connectives, the conditional and negation, defined by
\[\neg x=\fa\lef x\rig\tr.\] 
In~\cite{BP10a} we define various
\emph{short-circuit logics}: the fragments of proposition
algebra that remain if only $\leftand$ and $\neg$ can
be used. These logics (various choices can be made) are
put forward for modeling conditions as used in programming.
Typical laws that are valid with respect to 
each valuation congruence
are the associativity of $\leftand$, the double negation
shift, and $\fa\leftand x=\fa$ (and, as explained in the
Introduction, a typical non-validity 
is  $x\leftand \fa=\fa$).

\appendix
\section{Some proofs}
\label{app:Proofs}
\begin{la}[This is Lemma~\ref{lem:wm}, Section~\ref{sec:wmem}]
For each $t\in\NT_{\SigmaHMA}$ there exists a \wmem-basic form
$t'$ with $\CP_{\wmem}\vdash t=t'$.
\end{la}
\begin{proof}
By Lemma~\ref{lem:nf} we may assume that $t$ is a 
basic form and we proceed by structural induction on $t$. 
If $t\equiv \tr$ or $t\equiv \fa$ there is nothing to prove. 
If $t\equiv t_1\lef a\rig t_2$ we may assume that $t_i$ are
\wmem-basic forms (if not, they can proved equal to \wmem-basic forms). 
We first consider the positive side of $t$. If $a\not\in
pos(t_1)$ we are done, otherwise we saturate
$t_1$ by replacing each atom $b\ne a$ that occurs in a positive
position with $(a\lef b\rig\fa)$ using axiom \eqref{CPwm1}.
In this way we can retract 
each $a$ that is in $pos(t_1)$
(also using axiom \eqref{CPcr1}) and end up with $t_1'$ that does 
not contain $a$ on positive positions. For example,
\begin{align*}
t\equiv&(((\tr\lef a\rig R)\lef b\rig S)\lef c\rig V)\lef a \rig t_2\\
&~=(((\tr\lef a\rig R)\lef (a\lef b\rig\fa)\rig S)\lef
 (a\lef c\rig\fa)\rig V)\lef a \rig t_2\\
&~=(((((\tr\lef a\rig R)\lef a\rig S)\lef b\rig S)\lef
 a\rig V)\lef c\rig V)\lef a\rig t_2\\
&~=(((\tr\lef b\rig S)\lef a\rig V)\lef c\rig V)\lef a\rig t_2\\
&~=((\tr\lef b\rig S)\lef c\rig V)\lef a \rig t_2.
\end{align*}

Following the same procedure
for the negative side of $t$ (saturation with ($\tr\lef b\rig a)$ 
for all $b\ne a$ etc.) yields a \wmem-basic form
$t_1'\lef a\rig t_2'$ with $\CP_{\wmem}\vdash t=t_1'\lef a\rig t_2'$.
\end{proof}

\begin{la}[This is Lemma~\ref{lem:mem}, Section~\ref{sec:mem}]
For each $t\in\NT_{\SigmaHMA}$ there exists a \mem-basic form
$t'$ with $\CP_{\mem}\vdash t=t'$.
\end{la}

\begin{proof}
First observe that the axioms of $\CP_{\mem}$ imply 
the following simple consequences:
\begin{align}
\label{eq:mini3}
x\lef y\rig(v\lef y\rig w)&=x\lef y\rig w
&&\text{(take $u=\fa$ in axiom \eqref{CPmem})},
\\
\label{eq:mini5}
(x\lef y\rig z)\lef y\rig w&= x \lef y\rig w
&&\text{(take $u=\tr$ in equation~\eqref{eq:mini2})}.
\end{align}
By Lemma~\ref{lem:nf} we may assume that $t$ is a 
basic form and we proceed by structural induction on $t$. 
If $t\equiv\tr$ or $t\equiv\fa$ there is nothing to prove. 

Assume $t\equiv t_1\lef a\rig t_2$. 
We write $[\tr/a]t_1$ for the term that results when 
$\tr$ is substituted for 
$a$ in $t_1$. We first show that
\[\CP_{\mem}\vdash t_1\lef a\rig t_2=[\tr/a]t_1\lef a\rig t_2\]
by induction on $t_1$: if $t_1$ equals $\tr$ or $\fa$ this is clear.
If $t_1\equiv t_1'\lef a\rig t_1''$ then 
$\CP\vdash [\tr/a]t_1=[\tr/a]t_1'$ and we derive
\begin{align*}
t_1\lef a\rig t_2&=(t_1'\lef a\rig t_1'')\lef a\rig t_2\\
&=([\tr/a]t_1'\lef a\rig t_1'')
\lef a\rig t_2&&\text{by IH}\\
&=[\tr/a]t_1'\lef a\rig t_2&&\text{by \eqref{eq:mini5}}\\
&=[\tr/a]t_1\lef a\rig t_2,
\end{align*}
and if $t_1\equiv t_1'\lef b\rig t_1''$ with $b\neq a$ then $\CP\vdash
[\tr/a]t_1=[\tr/a]t_1'\lef b\rig [\tr/a]t_1''$ and we derive
\begin{align*}
t_1\lef a\rig t_2&=(t_1'\lef b\rig t_1'')\lef a\rig t_2\\
&=((t_1'\lef a\rig \tr)
\lef b\rig(t_1''\lef a\rig\tr))
\lef a\rig t_2&&\text{by \eqref{eq:mini1} and \eqref{eq:mini2}}\\
&=(([\tr/a]t_1'\lef a\rig \tr)
\lef b\rig([\tr/a]t_1''\lef a\rig\tr))
\lef a\rig P_2&&\text{by IH}\\
&=([\tr/a]t_1'\lef b
\rig [\tr/a]t_1'')
\lef a\rig t_2&&\text{by \eqref{eq:mini1} and \eqref{eq:mini2}}\\
&=[\tr/a]t_1\lef a\rig t_2.
\end{align*}
In a similar way, but now using \eqref{eq:mini3}, axiom 
\eqref{CPmem} and \eqref{eq:mini} instead, we find 
$\CP_{\mem}\vdash t_1\lef a\rig t_2=t_1\lef a\rig [\fa/a]t_2$, and thus
\[\CP_{\mem}\vdash t_1\lef a\rig t_2=[\tr/a]t_1\lef a\rig [\fa/a]t_2.\]
With axioms \eqref{CP1} and \eqref{CP2} 
we find basic forms $Q_i$ in which $a$ does not occur with 
$\CP_{\mem}\vdash Q_1=[\tr/a]P_1$ and $\CP_{\mem}\vdash 
Q_2=[\fa/a]P_2$. 

By induction it follows that there are \mem-basic
forms $R_1$ and $R_2$ with $\CP_\mem\vdash R_i=Q_i$, and hence
$\CP_\mem\vdash P=R_1\lef a \rig R_2$ and $R_1\lef a \rig R_2$ is a
\mem-basic form. 
\end{proof}

Before proving Lemma~\ref{lem:mem2} (Section~\ref{sec:mem}), 
we first formulate another lemma:
\begin{lemma}
\label{lem:mem1}
For all $a\in A$, $f\in M$, 
$t,t'\in\NT_{\SigmaHMA}$,
and $\rho\in (A\setminus\{a\})^{core}\cup\{\epsilon\}$,
\begin{equation}
\label{eq:pr}
(t'\apply(t\apply (a\apply f)))(\rho a)=
(t\apply (a\apply f))(\rho' a)
\end{equation}
for some $\rho'\in (A\setminus\{a\})^{core}
\cup\{\epsilon\}$.
\end{lemma}

\begin{proof}
By structural induction on $t'$.

If $t'\in\{\tr,\fa\}$ then \eqref{eq:pr} follows immediately.

If $t'\equiv a$ then $(a\apply(t\apply (a\apply f)))(\rho a)=
(t\apply (a\apply f))(a)$. 
\\\indent
Note that this case also covers $A=\{a\}$.

If $t'\equiv b\not\equiv a$ then $(b\apply(t\apply (a\apply f)))(\rho a)=
(t\apply (a\apply f))(b(\rho-b)a)$.

If $t'\equiv t_1\lef t_2\rig t_3$ we make a case distinction:
\begin{quote}
$(i)$ $t_2\reply (t\apply (a\apply f))=\tr$. Then 
\begin{align*}
(t'\apply(t\apply (a\apply f)))(\rho a)&=
(t_1\apply(t_2\apply(t\apply (a\apply f))))(\rho a)\\
&=
(t_1\apply((t_2\lef t\rig t_2)\apply (a\apply f)))(\rho a)\\
&=
((t_2\lef t\rig t_2)\apply (a\apply f))(\rho' a)&&\text{(by IH)}\\
&=
(t_2\apply(t\apply (a\apply f)))(\rho' a)\\
&=
(t\apply (a\apply f))(\rho'' a).&&\text{(by IH)}
\end{align*}

$(ii)$ $t_2\reply (t\apply (a\apply f))=\fa$. Similar.
\end{quote}
\end{proof}

\begin{la}[This is Lemma~\ref{lem:mem2}, Section~\ref{sec:mem}]
For all $f\in M$ and
$t,t'\in\NT_{\SigmaHMA}$,
\begin{equation*}
\tag{\ref{eq:prof}}
t\reply(t'\apply(t \apply f))=t\reply f~\wedge~ 
t\apply(t'\apply(t \apply f))=t'\apply(t\apply f).
\end{equation*}
\end{la}

\begin{proof}
We prove this property by structural induction on $t$.
\\[1mm]
If $t\in\{\tr,\fa\}$ then \eqref{eq:prof} follows immediately.
\\[1mm]
If $t\equiv a\in A$ then apply Lemma~\ref{lem:mem1} with
$t=\tr$ and derive
$(t'\apply a\apply f)(\rho a)=(a\apply f)(\rho' a)=f(a)$,
thus
$a\reply(t'\apply(a \apply f))=
(t'\apply(a \apply f))(a)=f(a)=a\reply f$.
Furthermore, 
$a\apply(t'\apply(a \apply f))=
t'\apply(a\apply f)$ follows by structural induction on $t'$:
\begin{quote}
$t'\equiv\tr$:
$
(a\apply(a\apply f))(\sigma)=
\begin{cases}
(a\apply f)(a)=f(a)&\text{if $\sigma=a$ or $\sigma=\rho a$},\\
(a\apply f)(a(\sigma-a))=f(a(\sigma-a))&
\text{otherwise},
\end{cases}
$
\\
thus $a\apply (a\apply f)=a\apply f$,

$t'\equiv\fa$: similar,

$t'\equiv a$: similar,

$t'\equiv b\not\equiv a$ then consider both functions applied
to $\rho\in A^{core}$: 

$(i)$ if $\rho$ ends with $a$
then by definition both 
functions yield $f(a)$, 

$(ii)$ if $\rho$ ends with $b$ then
$(a\apply(b\apply (a\apply f)))(\rho)= (b\apply (a\apply f))(a(\rho-a))
=(a\apply f)(b)=f(ab)$
and $(b\apply (a\apply f))(\rho)=(a\apply f)(b)=f(ab)$,

$(iii)$ in the remaining case
$\rho$ does not end with either $a$ or $b$, so
\\
$(a\apply(b\apply(a \apply f)))(\rho)=
(b\apply(a \apply f))(a(\rho-a))=
(a \apply f)(ba((\rho-a)-b))=f(ab((\rho-a)-b))
$
and
$(b\apply(a \apply f))(\rho)=
(a \apply f)(b(\rho-b))=
f(ab((\rho-b)-a))
$, 
\\
so both functions are the same,

$t'\equiv t'_1\lef t'_2\rig t'_3$ and we make a case distinction:

$(i)$ if $t'_2\reply(a\apply f)=\tr$ then we find by IH that
\begin{align*}
t'\apply(a\apply f)&=
t'_1\apply(t'_2\apply(a\apply f))\\
&=
t'_1\apply(a\apply(t'_2\apply(a\apply f)))\\
&=
a\apply(t'_1\apply(a\apply(t'_2\apply(a\apply f))))
\\
&=a\apply(t'_1\apply(t'_2\apply(a\apply f))),
\end{align*}
$(ii)$ $t'_2\reply(a\apply f)=\fa$. Similar.
\end{quote}
If $t \equiv t_1\lef t_2\rig t_3$ we make
a case distinction:
\\[1mm]
$(i)$ $t_2\reply f=\tr$. By IH we find
$t_2\reply(t'\apply(t_1\apply (t_2\apply f)))=
t_2\reply((t'\lef t_1\rig t')\apply (t_2\apply f))=
t_2\reply f=\tr$ and we derive
\begin{align*}
t\reply(t'\apply(t \apply f))&=t\reply(t'\apply(t_1\apply (t_2\apply f)))
\\
&=(t_1\lef t_2\rig t_3)\reply(t'\apply(t_1\apply (t_2\apply f)))\\
&=t_1\reply(t_2\apply(t'\apply(t_1\apply (t_2\apply f))))
\\
&=t_1\reply((t_2\lef t'\rig t_2)\apply(t_1\apply (t_2\apply f)))
\\
&=t_1\reply(t_2\apply f)&&\text{(by IH)}\\
&=t\reply f,
\end{align*}
and
\begin{align*}
t\apply(t'\apply(t \apply f))&=t\apply(t'\apply(t_1\apply (t_2\apply f)))
\\
&=(t_1\lef t_2\rig t_3)\apply(t'\apply(t_1\apply (t_2\apply f)))\\
&=t_1\apply(t_2\apply(t'\apply(t_1\apply (t_2\apply f))))
\\
&=t_1\apply((t_2\lef t'\rig t_2)\apply(t_1\apply (t_2\apply f)))
\\
&=t'\apply(t_1\apply(t_2\apply f))&&\text{(by IH)}\\
&=t'\apply(t\apply f).
\end{align*}
$(ii)$ $t_2\reply f=\fa$. Similar.
\end{proof}

\end{document}